\documentclass{article}
\usepackage{fullpage}

\title{New Deterministic Algorithms for Solving Parity Games\thanks{This research was supported by ERC Starting Grant 306465 (BeyondWorstCase).}}

\newcommand*\samethanks[1][\value{footnote}]{\footnotemark[#1]}

\author{Matthias Mnich\thanks{Department of Computer Science, University of Bonn, Germany. \texttt{\{mmnich@,roeglin@cs.,roesner@cs.\}uni-bonn.de}}
   \and Heiko R{\"o}glin\samethanks[2]
   \and Clemens R{\"o}sner\samethanks[2]
   }

\usepackage{todonotes}
\usepackage{tikz}
\usetikzlibrary{calc}
\usetikzlibrary{decorations.markings,decorations.pathmorphing}
\usepackage{hyperref}
\usepackage{booktabs}
\usepackage{enumitem}
\usepackage{color}
\usepackage{amssymb,amsmath,amsthm}
\usepackage{caption,subcaption}
\usepackage{algorithm}
\usepackage[noend]{algpseudocode}
\usepackage{xspace}

\newcommand{\reach}{\mathsf{reach}}

\newtheorem{theorem}{Theorem}
\newtheorem{lemma}{Lemma}
\newtheorem{corollary}{Corollary}

\begin{document}

\maketitle

\begin{abstract}
  We study parity games in which one of the two players controls only a small number~$k$ of nodes and the other player controls the $n-k$ other nodes of the game.
  Our main result is a fixed-parameter algorithm that solves bipartite parity games in time~$k^{O(\sqrt{k})}\cdot O(n^3)$, and general parity games in time\linebreak $(p+k)^{O(\sqrt{k})} \cdot O(pnm)$, where~$p$ is the number of distinct priorities and~$m$ is the number of edges.
  For all games with $k = o(n)$ this improves the previously fastest algorithm
by Jurdzi{\'n}ski, Paterson, and Zwick (SICOMP 2008).
  
  \quad We also obtain novel kernelization results and an improved deterministic algorithm for 
graphs with small average degree.
\end{abstract}

\section{Introduction}
\label{sec:introduction}
A parity game~\cite{EmersonJutna1991} is a two-player game of perfect information played on a directed graph $G$ by two players, \emph{even} and \emph{odd}, who move a token from node to node along the edges of $G$ so that an infinite path is formed.
The nodes of $G$ are partitioned into two sets $V_0$ and $V_1$; the even player moves if the token is at a node in $V_0$ and the odd player moves if the token is at a node in $V_1$.
The nodes of $G$ are labeled by a \emph{priority function} $p: V\rightarrow\mathbb N_0$, and the players compete for the parity of the highest priority occurring infinitely often on the infinite path $v_0,v_1,v_2\hdots$ describing a play: the even player wins if $\limsup_{i\rightarrow\infty}p(v_i)$ is even, and the odd player wins if it is odd.

The winner determination problem for parity games is the algorithmic problem to determine for a given parity game $G = (V_0\uplus V_1,E,p)$ and an initial node $v_0\in V_0\cup V_1$, whether the even player has a winning strategy in the game if the token is initially placed on node $v_0$.
We say that an algorithm for this problem \emph{solves} parity games.
Parity games have various applications in computer science and the theory of formal languages and automata in particular.
They are closely related to other games of infinite duration, such as mean payoff games, discounted payoff games, and stochastic games~\cite{Jurdzinski1998}.
Solving parity games is linear-time equivalent to the model checking problem for the modal $\mu$-calculus~\cite{Stirling1995}.
Hence, any parity game solver is also a model checker for the $\mu$-calculus (and vice versa).

Many algorithms have been suggested for solving parity games~\cite{BjorklundEtAl2003,JurdzinskiEtAl2008,JurdzinskiVoge2000,Zielonka1998}, yet none of them is known to run in polynomial time.
McNaughton~\cite{McNaughton1993} showed that the winner determination problem belongs to the class \linebreak $\mathsf{NP}\cap \mathsf{coNP}$, and Jurdzi{\'n}ski~\cite{Jurdzinski1998} strengthened this to $\mathsf{UP}\cap \mathsf{coUP}$.
It is a long-standing open question whether parity games can be solved in polynomial time.
The fastest known deterministic algorithm is due to Jurdzi{\'n}ski, Paterson, and Zwick~\cite{JurdzinskiEtAl2008}
and it has a run time of~$n^{O(\sqrt{n})}$ for general parity games and of~$n^{O(\sqrt{n/\log n})}$ for parity 
games in which every node has out-degree at most two.
The fastest known randomized algorithm for general parity games is due to Bj{\"o}rklund et al.~\cite{BjorklundEtAl2003}
and it has a run time of~$n^{O(\sqrt{n/\log n})}$.

As a polynomial-time algorithm for solving parity games has remained elusive,
researchers have started to consider which restrictions on the game allow for polynomial-time algorithms.
One such well-studied restriction is the treewidth~$t$ of the underlying undirected graph $G$ of the game.
Obdr{\v z}{\'a}lek~\cite{Obdrzalek2003} found an algorithm solving parity games on $n$ nodes in time $n^{O(t^2)}$.
Later, Fearnley and Lachish~\cite{FearnleyLachish2011} gave an algorithm solving parity games in time $n^{O(t\log n)}$.
Another well-studied parameter for parity games is the number~$p$ of distinct priorities by which the nodes of the game are labeled.
The progress-measure lifting algorithm by Jurdzi{\'n}ski~\cite{Jurdzinski2000} solves parity games in time $O(pm(2n/p)^{p/2})$, where $m$ denotes the number of edges of~$G$.
This run time has been improved by Schewe~\cite{Schewe2007} to $O(m((2e)^{3/2}n/p)^{p/3})$.
Fearnley and Schewe~\cite{FearnleySchewe2012} presented an algorithm for solving parity games with
run time $O(n (t+1)^{t+5} (p+1)^{3t+5})$, assuming that a tree decomposition of~$G$ with width~$t$ is given.

For a given parameter~$\kappa$, one usually aims for \emph{fixed-parameter algorithm} algorithms, i.e., algorithms that run in time $f(\kappa)\cdot n^c$ for some computable function~$f$ and some constant~$c$ that is independent of~$\kappa$.
Such an algorithm can be practical for large instances if $f$ grows moderately and $c$ is small.
From the previously mentioned algorithms only the algorithm by Fearnley and Schewe~\cite{FearnleySchewe2012} is a fixed-parameter algorithm for the combined parameter~$(t,p)$.
It is not known if fixed-parameter algorithms exist for the parameter~$t$ or the parameter~$p$ alone.   

Further parameters for which polynomial-time algorithms for parity games have been suggested include DAG-width~\cite{BerwangerEtAl2006}, clique-width~\cite{Obdrzalek2007}, and entanglement~\cite{BerwangerEtAl2012}; none of these are fixed-parameter algorithms.

\subsection{Our Contributions}
\label{sec:ourcontributions}
We study as parameter the number~$k$ of nodes that belong to the player who controls the smaller number of nodes in the parity game.
Our first result is a \emph{subexponential} fixed-parameter algorithm for solving general parity games for parameters~$p$ and~$k$ and for parameter only~$k$ for bipartite parity games (where players alternate between their moves).
\begin{theorem}
\label{thm:fixed-parameter-k-main}
  There is a deterministic algorithm that solves any parity game~$G$ on $n$ nodes and $m$ edges in time $(p+k)^{O(\sqrt{k})} \cdot O(pnm)$, where~$k$ denotes the minimum number of nodes owned by one of the players and~$p$ the number of distinct priorities.
  If~$G$ is bipartite, the algorithm runs in time $k^{O(\sqrt{k})} \cdot O(n^3)$.
\end{theorem}
Thus, our algorithm is particularly efficient if the game is unbalanced, in the sense that one player owns only~$k$ nodes and the other player owns the remaining $n-k \gg k$ nodes.

Let us remark that it is not very hard to show fixed-parameter tractability for parameter $p+k$; indeed McNaughton's algorithm~\cite{McNaughton1993} can be shown to run in time $p^k\cdot n^{O(1)}$, and this was improved to $p^{\log k}\cdot 4^k\cdot n^{O(1)}$ by Gajarsk{\'y} et al.~\cite{GajarskyEtAl2015}.
Our key contribution here is to reduce the dependence of $k$ to a \emph{subexponential} function.
Indeed, this improvement allows us to derive the following immediate corollary of Theorem~\ref{thm:fixed-parameter-k-main} to expedite the run time for solving \emph{general} parity games.
\begin{corollary}
  There is a deterministic algorithm that solves parity games in time~$n^{O(\sqrt{k})}$.
\end{corollary}
Our algorithm is asymptotically always at least as fast as the fastest known deterministic parity game solver by Jurdzi{\'n}ski, Paterson, and Zwick~\cite{JurdzinskiEtAl2008}, which runs in time $n^{O(\sqrt{n})}$.
For the case~$k= o(n)$, our algorithm is asymptotically faster than theirs and constitutes the fastest known deterministic solver for such games.

We also prove the existence of a small kernel, as our second result.
For a parameterized problem, a \emph{kernelization algorithm} takes as input an instance $x$ with parameter $\kappa$ and computes in time $(|x| + \kappa)^{O(1)}$ an equivalent instance~$x'$ with parameter $\kappa'$ (a \emph{kernel}) with size $|x'| \leq g(\kappa)$, for some computable function~$g$; here, equivalent means that an optimal solution for $x$ can be derived in polynomial time from an optimal solution of $x'$.
\begin{theorem}
\label{thm:kernel-pk-main}
  Parity games can be kernelized in time $O(pmn)$ to at most $(p+1)^k+(p+1) k$ nodes, and bipartite parity games can be kernelized in time
$O(n^3)$ to at most~$k+2^k\cdot\min\{k,p\}$ nodes and at most $k2^k\cdot\min\{k,p\}$ edges.
\end{theorem}
This kernelization result is not only interesting for its own sake, but it is also an important ingredient in the proof of Theorem~\ref{thm:fixed-parameter-k-main}. 

As our third result, we generalize the algorithm by Jurdzi{\'n}ski, Paterson, and Zwick~\cite{JurdzinskiEtAl2008}
for parity games with maximum out-degree 2 to arbitrary out-degree~$\Delta$.
\begin{theorem}
\label{thm:outdegreebased-algorithm}
  There is a deterministic algorithm that solves parity games on $n$ nodes out of which $s_j$ nodes have out-degree at most $j$ in time 
  \begin{equation*}
   n^{O\left(\min_{1\leq j \leq n}\left\{\sqrt{n-s_j} + \sqrt{\frac{s_j}{\log_{j}{s_j}}}\right\}\right)} \enspace .
  \end{equation*}
\end{theorem}

\begin{corollary}
\label{cor:outdegreebased-algorithm}
  There is a deterministic algorithm that solves parity games on~$n$ nodes with maximum out-degree~$\Delta$ in time~$n^{O(\sqrt{\log(\Delta)\cdot n/\log(n)})}$ and parity games on~$n$ nodes with average out-degree~$\Delta$ in time $n^{O(\sqrt{\log(\log(n)\Delta)\cdot n/\log(n)})}$.
\end{corollary}

\noindent

\subsection{Detailed Comparison with Previous Work}

Let us discuss in detail how our results compare to previous work.  
It is well-known (cf.~\cite[Lemma 3.2]{KloksBodlaender1992}) and easy to prove that the treewidth of a complete bipartite graph equals the size of the smaller side.
Since the treewidth of a graph can only decrease when deleting edges, the graph underlying a bipartite parity game 
in which one player owns~$k$ nodes has a treewidth of at most~$k$. However, as it is not known if there exists
a fixed-parameter algorithm for parameter treewidth, the result in Theorem~\ref{thm:fixed-parameter-k-main} for the bipartite case
does not follow from previous work about parity games with bounded treewidth. As a parity game in which one player
owns~$k$ nodes can have up to~$n$ different priorities, also the fixed-parameter algorithm for the combined parameter~$(t,p)$ 
by Fearnley and Schewe~\cite{FearnleySchewe2012} does not imply our result.

The algorithm of Jurdzi{\'n}ski, Paterson, and Zwick~\cite{JurdzinskiEtAl2008}
for parity games with maximum out-degree two with run time~$n^{O(\sqrt{n/\log n})}$
can easily be generalized to arbitrary parity games at the expense of its run time.
For this, one only needs to observe that every parity game can be transformed into
a game with maximum out-degree two by replacing each node with a higher out-degree
by an appropriate binary tree. This transformation increases the number of nodes from~$n$
to~$\Theta(m)$ where~$m$ denotes the number of edges in the original parity game.
Hence, the run time becomes~$m^{O(\sqrt{m/\log m})}=n^{O(\sqrt{m/\log n})}$.
For graphs with average out-degree~$\Delta=\omega(\log\log{n})$ the resulting run time of~$n^{O(\sqrt{\Delta n/\log n})}$
is asymptotically worse than the run time we obtain in Corollary~\ref{cor:outdegreebased-algorithm}

For graphs in which the variance of the out-degrees is large, our algorithm can even be better than stated in Corollary~\ref{cor:outdegreebased-algorithm}.
If, for example, there are~$n^{1-\varepsilon}$ nodes with an arbitrary out-degree for some~$\varepsilon>0$ and all remaining nodes have constant out-degree
at most~$c$ then our algorithm has a run time of $n^{O(\sqrt{\frac{n}{\log{n}}})}$ (the minimum in Theorem~\ref{thm:outdegreebased-algorithm} is assumed for~$j=c$). 
This matches the best known bound for randomized algorithms.

Gajarsk{\'y} et al.~\cite{GajarskyEtAl2015} present an algorithm that solves parity games in time $w^{O(\sqrt{w})}\cdot n^{O(1)}$, where~$w$ denotes the modular width of~$G$. Since the modular width of a bipartite graph can be exponential in the size of the smaller side, Theorem~\ref{thm:fixed-parameter-k-main} does not follow from this result.

\section{Fundamental Properties of Parity Games}
\label{sec:paritygames}
A parity game $G = (V_0 \uplus V_1,E,p)$ consists of a directed graph $(V_0\uplus V_1,E)$, where~$V_0$ is the set of \textit{even} nodes and $V_1$ is the set of \textit{odd} nodes, and a priority function $p: V_0 \cup V_1 \rightarrow \mathbb{N}_0$.
We often abuse notation and also refer to $(V_0\uplus V_1,E)$ as the graph $G$.
For each node $v\in V(G)$, we denote by $N_G^+(v)$ and $N_G^-(v)$ the set of out-neighbors and in-neighbors of $v$ in $G$, respectively.

Two standard assumptions about parity games are (1) that $G$ is bipartite with $E \subseteq (V_0 \times V_1) \cup (V_1 \times V_0)$, and (2) that each node $u \in V$ has at least one outgoing edge $(u,v) \in E$.
The first assumption is often made because it is easy to transform a non-bipartite instance into a bipartite instance.
However, the usual transformation increases the number of nodes in~$V_i$ by an amount of $|\{v \in  V_{1-i}\mid  N_G^-(v) \cap V_{1-i} \neq \emptyset\}|$, and can therefore increase the parameter $k = \min\{|V_0|,|V_1|\}$ significantly.
We therefore consider bipartite and non-bipartite instances separately in Theorem~\ref{thm:fixed-parameter-k-main}.

We write $n = |V(G)|$, $m = |E|$ and $p = |\{p(v)\mid v \in V(G)\}|$.
The game is played by two players, the \textit{even} player (or player~$0$) and the \textit{odd} player (or player~$1$).
The game starts at some node $v_0 \in V(G)$.
The players construct an infinite path (a \emph{play}) as follows.
Let $u$ be the last node added so far to the path.
If $u \in V_0$, then player~$0$ chooses an edge $(u, v) \in E$.
Otherwise, if $u \in V_1$, then player $1$ chooses an edge $(u, v) \in E$.
In either case, node $v$ is added to the path and a new edge is then chosen by either player~$0$ or player~$1$.
As each node has at least one outgoing edge, the path constructed can always be continued.
Let $v_0, v_1, v_2, \ldots$ be the infinite path constructed by the two players and let $p(v_0), p(v_1), p(v_2), \ldots$ be the sequence of the priorities of the nodes on the path.
Player $0$ \emph{wins} the game if the largest priority seen infinitely often is even, and
player~$1$ wins if the largest priority seen infinitely often is odd.

We will define $p_1(v)$ as $p(v)$ if $p(v)$ is odd and as $-p(v)$ if $p(v)$ is even.
This allows us to say that, in case $p_1(v) > p_1(u)$ for some $v,u \in V$, player~$1$ \emph{prefers}~$p(v)$ over $p(u)$.
Observe that removing an arbitrary finite prefix of a play in a parity game does not change the winner; we refer to this property of parity games as \emph{prefix independence}.
A \emph{strategy} for player $i\in\{0,1\}$ in a game~$G$ specifies, for every finite path $v_0, v_1,\ldots, v_k$ in $G$ that ends in a node $v_k \in V_i$, an edge $(v_k, v_{k+1}) \in E$.
A strategy is \emph{positional} if the edge $(v_k, v_{k+1}) \in E$ chosen depends only on the last node $v_k$ visited and is independent of the prefix path $v_0, v_1,\ldots, v_{k-1}$.
A strategy for player $i\in \{0,1\}$ is \emph{winning} (for player $i$) from a start node~$v_0$ if following this strategy ensures that player $i$ wins the game, regardless of which strategy is used by the other player.

The fundamental determinacy theorem for parity games \cite{EmersonJutna1991,GradelEtAl2002} says that for every parity game $G$ and every start node $v_0$, either player $0$ has a winning strategy or player $1$ has a winning strategy.
Furthermore, if a player has a winning strategy from some node in a parity game, then she also has a winning positional strategy from this node.
From now on we will therefore, unless stated differently, assume every strategy to be positional.
Given positional strategies $s_0$ on $V_0$ and~$s_1$ on~$V_1$ and a start node $v_0 \in V$ the infinite path starting in~$v_0$ corresponding to these strategies consists of a finite prefix and an infinite recurrence of a cycle $C = C(s_0,s_1,v_0)$.
We call $C$ the cycle \emph{corresponding to $s_0,s_1,v_0$} and say that $s_0$ and $s_1$ \emph{create} $C$.
The parity of the highest priority $p(u)$ of all nodes $u\in V(C)$ in cycle $C$ then determines the winner of the game.
The \emph{winning set} of player~$i \in \{0,1\}$ is the set $\mathsf{win}_i(G)\subseteq V$ of nodes of the game~$G$ from which player~$i$ has a winning strategy.

For $i \in \{0,1\}$, an \emph{$i$-dominion} is a set of nodes $D \subseteq V$ so that player~$i$ can win from every node of $D$, without leaving $D$ and without allowing the other player to leave~$D$.
An example of an $i$-dominion is the set $\mathsf{win}_i(G)$, but there may be smaller subsets of $\mathsf{win}_i(G)$ that are $i$-dominions as well.
Although finding $i$-dominions can be just as hard as finding $\mathsf{win}_i(G)$, searching only for dominions with certain properties (e.g. small dominions) can be easier.
In our algorithm we will use the fact that once an $i$-dominion is found, it can easily be removed from the graph, leaving a smaller game to be solved.

\medskip

Next, we recall some well-known results about parity games that form the basis of the algorithms for solving parity games by McNaughton~\cite{McNaughton1993} and Zielonka~\cite{Zielonka1998}.
We include them here as our algorithm relies on them as well; for a detailed exposition we refer to Gr{\"a}del et al.~\cite{GradelEtAl2002}.
Fix a parity game $G = (V_0\uplus V_1,E,p)$.

For $i\in\{0,1\}$, a set $B \subseteq V(G)$ is \emph{$i$-closed} if for every $u \in B$ the following holds (we use the notation $\neg i$ for the element $1-i\in\{0,1\}$):
\begin{itemize}
 \item If $u \in V_i$, then there exists some $(u,v) \in E$ such that $v \in B$; and
 \item if $u \in V_{\neg i}$, then for every $(u,v) \in E$, we have $v \in B$.
\end{itemize}
In other words, a set $B$ is $i$-closed if player $i$ can always choose to stay in $B$ while simultaneously player $\neg i$ cannot escape from it, i.e., $B$ is a ``trap'' for player $\neg i$.

\begin{lemma}
\label{thm:winning-sets-are-closed}
 For each $i \in \{0,1\}$, the set $\mathsf{win}_i(G)$ is $i$-closed.
\end{lemma}

Let $A \subseteq V(G)$ be a set of nodes and let $i\in\{0,1\}$.
The \emph{$i$-reachability set} of $A$ is the set $\reach_i(A)$ of nodes in $A$ together with all nodes in $V(G)\setminus A$ from which player~$i$ has a strategy $\sigma$ to enter $A$ at least once (regardless of the strategy of the other player); we call such a strategy $\sigma$ an \emph{$i$-reachability strategy} to~$A$.

\begin{lemma}
\label{thm:winning-sets-are-complement-closed}
 For $A \subseteq V(G)$ and $i \in \{0,1\}$, the set $V(G) \setminus \reach_i(A)$ is $(\neg i)$-closed.
\end{lemma}

We will from now on assume that the graph of the parity game we operate on is encoded as an adjacency list.

\begin{lemma}
\label{thm:computing-reach-linear-time}
 For every set $A \subseteq V(G)$ and $i \in \{0, 1\}$, the set $\reach_i(A)$ can be computed in $O(m)$ time, where $m = |E|$ is the number of edges in the game.
\end{lemma}

If $B \subseteq V(G)$ is such that for each node $u \in V(G) \setminus B$ there is an edge $(u, v)$ with $v \in V(G) \setminus B$, then the sub-game $G - B$ is the game obtained from $G$ by removing the nodes of $B$.
We will only be using $B$'s for which $V(G) \setminus B$ is an $i$-closed set for some $i$.
In this case every node in $v \in V(G) \setminus B$ has at least one out-going edge $(v,w)$ with $w \in V(G) \setminus B$ and $G-B$ will therefore be well-defined. 
The next lemmas show some useful properties of sub-games.
\begin{lemma}
\label{thm:winning-sets-closed-under-subgames}
 Let $G'$ be a sub-game of $G$ and let $i \in \{0, 1\}$.
 If the node set of $G'$ is $i$-closed in $G$, then $\mathsf{win}_i(G') \subseteq \mathsf{win}_i(G)$.
\end{lemma}

The next lemma shows that if we know some non-empty subset $U$ of the winning set of some player $\neg i$ in a game $G$, then computing the winning sets of both players in $G$ can be reduced to computing their winning sets in the smaller game $G - \reach_{\neg i}(U)$.
\begin{lemma}
\label{thm:winning-set-from-reachability-set}
 For any parity game $G$ and $i \in \{0,1\}$, if $U \subseteq \mathsf{win}_{\neg i}(G)$ and $U^* = \reach_{\neg i}(U)$, then $\mathsf{win}_{\neg i}(G) = U^* \cup \mathsf{win}_{\neg i}(G - U^*)$ and $\mathsf{win}_i(G) = \mathsf{win}_i(G - U^*)$.
\end{lemma}

The next lemma complements Lemma~\ref{thm:winning-set-from-reachability-set} by providing a way to find a non-empty subset of the winning set of player $\neg i$ in a parity game~$G$ or to conclude that player $i$ can win from every node in $G$.
\begin{lemma}
\label{thm:recurse-after-reach-removable}
 Let $G$ be a parity game with largest priority $p_{\max}$ and let $V_{p_{\max}} \subseteq V(G)$ be the set of nodes with priority $p_{\max}$.
 Let $i = p_{\max} \pmod 2$ and let $G' = G - \reach_i(V_{p_{\max}})$.
 Then $\mathsf{win}_{\neg i}(G') \subseteq \mathsf{win}_{\neg i}(G)$.
 Also, if $\mathsf{win}_{\neg i}(G') = \emptyset$, then $\mathsf{win}_i(G) = V$, i.e., player~$i$ wins from every node of $G$.
\end{lemma}

\section{Kernelization of Parity Games}
\label{sec:reductionofparitygames}
In this section, we describe some reduction rules for parity games.
Theses rules are such that we can efficiently compute the winning sets of the original parity game once we know the winning sets of the reduced game.

\subsection{General Parity Games}
\begin{lemma}
\label{lemma:reduction}
  Any parity game $G = (V_0 \uplus V_1,E,p)$ can be transformed in time $O(pmn)$ to a parity game $G' = (V'_0 \uplus V'_1,E',p')$ with $V'_1\subseteq V_1$ such that 
  \begin{itemize}
    \item there are no edges inside~$V_1'$, and
    \item for each node $v \in V'_0$ either $N_G^+(v) \subseteq V'_1$ or $N_G^-(v) \subseteq V'_1$, and
    \item $|V'_0| \leq \min\{n + p k, (p+1)^k+p k\}$, where $k = |V_1|$.
  \end{itemize}
   Moreover, $G$ and $G'$ have the same winning sets on $V'_1$ and the winner of the remaining nodes of $G$ can be computed either during the transformation or from the winning sets of $G'$ in linear time.
\end{lemma}

\begin{proof}
  We will modify $G$ in multiple steps. We will slightly abuse notation and refer 
  in every step to the parity game that we obtained in the step before as~$G = (V_0 \uplus V_1,E,p)$.
  First we eliminate all edges inside~$V_1$.
  This can easily be achieved by adding for each edge $e = (v,w) \in E$ with $v,w \in V_1$
  a new node~$v_e$ with~$p'(v_e) = p(w)$ to $V_0$ and by replacing the edge~$e$ by the two edges~$(v,v_e)$ and~$(v_e,w)$.
  Since the new node~$v_e$ has only a single outgoing edge, this transformation does neither change the winning sets
  nor the winning strategies.

  Next, we remove certain cycles inside~$V_0$ from the game. Let~$W_0\subseteq V_0$ denote all nodes in~$V_0$ that
  are part of at least one cycle that lies completely inside~$V_0$ and whose highest priority is even. Clearly player~0
  can win from all nodes in~$\reach_0(W_0)$ by enforcing that such a cycle is entered and never left again. 
  Hence, we can remove~$\reach_0(W_0)$ from the game according to Lemma~\ref{thm:winning-set-from-reachability-set}. 
  Let~$W_1\subseteq V_0$ denote all nodes that are left in~$V_0$ and from which player~0 cannot reach~$V_1$.
  Then all paths that start in some node $u\in W_1$ must end in some cycle that is completely contained in $V_0$.
  Since we have removed all cycles whose highest priority is even, the maximum priority of this cycle must be odd.
  Thus, player~1 wins from all nodes in~$\reach_1(W_1)$. Hence, we can also remove~$\reach_1(W_1)$ from the game according to Lemma~\ref{thm:winning-set-from-reachability-set}.

  We use again the notation~$G = (V_0 \uplus V_1,E,p)$ to refer to the parity game obtained after the previously discussed steps. 
  Since we have removed all cycles from~$V_0$ whose highest priority is even, player~$0$ loses for sure if she does
  not leave~$V_0$. Hence, we can assume without loss of generality that the play leaves~$V_0$ from every starting
  node if player~$0$ plays an optimal strategy. Then for every node~$v\in V_0$ player~0 uses a (possibly empty) path inside~$V_0$
  followed by an edge that leads to some node~$w \in V_1$. 
  To determine the winning sets of a strategy of player~$0$ it is not important to know the exact paths player~$0$ chooses.
  Rather, it suffices to know for each $v \in V_0$ which node $w \in V_1$ will be reached and what the highest priority on the chosen $v$-$w$-path is.
  To get rid of long paths, we add $p \cdot |V_1|$ new nodes to~$V_0$, one node $v(p',w)$ for each pair of a priority $p'$ and a node $w\in V_1$.
  Node $v(p',w)$ has a priority $p'$ and its only out-neighbor is $w$.
  The winner does not change if player~$0$ goes from~$v\in V_0$ directly to~$v(p',w)$ and from there directly to~$w\in V_1$
  instead of taking some other path from~$v$ inside~$V_0$ with maximum priority~$p'$, followed by an edge that leads to~$w$.
	For all such paths we add the corresponding edge~$(v,v(p',w))$ and can therefore delete all edges inside $V_0$ that do not end in one of the new nodes~$v(p',w)$
  without changing the winning sets of the game. Observe that this ensures that all out-neighbors of the 
  new nodes~$v(p',w)$ belong to~$V_1$ while all in-neighbors of the old nodes~$v\in V_0$ belong to~$V_1$. 

  It can be the case that for some pair~$(v,w)\in V_0\times V_1$ there are multiple nodes~$v(p',w)$ that can be
  reached from~$v$. We can assume without loss of generality that if player~$0$ decides to go from~$v$ to~$w$ via one of these nodes then
  she chooses the one that is best for her, i.e., the one with lowest~$p_1$-value. All edges from~$v$ to other nodes~$v(p',w)$
  can be removed.    

  The inequality~$|V'_0| \leq n+\min\{m,k^2\}+pk$ follows directly from the previously discussed construction:
  initially~$V_0$ consists of~$n-k\le n$ nodes, there are at most~$\min\{m,k^2\}$ edges inside~$V_1$ for which we create a new node~$v_e$,
  and there are only~$pk$ new nodes~$v(p',w)$. To get rid of the term~$\min\{m,k^2\}$ we can identify each node~$v_e$, which derived from an edge~$e=(v,w)$ inside~$V_1$, with the node~$v(p(w),w)$. 
  This ensures that there are only~$pk$ new nodes.
	To show that $|V'_0| \leq (p+1)^k + pk$ we can reduce the number of old nodes in $V_1$ to ensure that at most $(p+1)^k$ remain.
  At first we remove all nodes $v \in V_0$ with $N_G^-(v) = \emptyset$, because they obviously cannot be part of a cycle and we can compute in linear time to which winning set they belong, once we know to which winning set their out-neighbors belong. Now let $v$ and $v'$ be two such nodes in $V_0$ with $N_G^+(v) = N_G^+(v')$.
  We then identify $v$ and $v'$ without changing the winning sets in $V_1$, since all nodes in $N_G^+(v)$ must have a priority at least as high as $\max\{p(v),p(v')\}$. This is because the priority of any node in $N_G^+(v)$ ($N_G^+(v')$) corresponds to the highest priority on a path that starts in $v$ ($v'$) and therefore must be at least $p(v)$ ($p(v')$).
  Afterwards there remains at most one node $v\in V_0$ for each possible set $N_G^+(v)$.

  Since $N_G^+(v)$ can contain at most one new node corresponding to $w$ for each $w \in V_1$ and there are $p$ different ones to choose from there are at most $(p+1)^k$ different possibilities for $N_G^+(v)$.

  It remains to analyze the run time of the transformation.
  We consider the different steps of the reduction separately.
  The first step of removing all edges inside $V_1$ can be performed in $O(m)$ because we only need to check for every edge $e=(v,w) \in E$ if~$v,w \in V_1$ and then remove one edge and add two edges and a node.
  The second step of removing dominions completely inside $V_0$ can be executed in time $O(\log(p) \cdot m)$ as follows.
  First, we solve the solitary game on $V_0 \setminus \reach_1(V_1)$ and remove the $0$-reachability-set of its $0$-winning set; this can be done in time $O(\log{p} \cdot m)$ \cite{BerwangerEtAl2004}.
  Thereafter, we compute the $1$-reachability set of $V_0 \setminus \reach_0(V_1)$ and remove it; this can be done in time $O(m)$~\cite{JurdzinskiEtAl2008}.
  The third step of removing long paths inside $V_0$ can be performed as follows.
  The algorithm computes the best priority for player~$0$ that a path in $V_0$ from a node $v \in V_0$ to a node $w \in V_1$ can have.
  To determine which nodes in $v\in V_0$ can reach which nodes $w \in V_1$ via a path in $V_0$ whose highest priority has fixed value of $p'$, consider the subgraph $G^{\leq p'}$ of $G$ that is induced by the set~$V^{\leq p'}$ of nodes of priority at most $p'$ and remove from it those edges that start in $V_1$.
  We then consider the set of nodes with priority exactly $p'$ and compute by DFS in time~$O(m)$ all nodes in $V_1$ reachable from them.
  Then we compute with DFS for each node in $V_0$ which of the nodes with priority $p'$ they can reach.
This takes a total of at most $2n$ applications of DFS for each priority and therefore in total $O(p m n)$ time.
  In the last step where we remove and contract some of the nodes in~$V_0$ we can find all nodes without incoming edges in time $O(m)$ and we can order all remaining nodes by their outgoing edges in time  $O(|V_1| \cdot (n + p+1))$ using a version of radix-sort, where we view the set of out-neighbors as an $(p+1)$-adic number with $|V_1|$ digits.
  Thereafter, in linear time we identify sets of nodes with the same outgoing neighbors and identify them in total time $O(n+m)$.
\end{proof}

\subsection{Bipartite Parity Games}

In this section we give some reduction rules that efficiently reduce any bipartite game $G = (V_0 \uplus V_1,E,p)$ to a structurally simpler bipartite game $G'=(V'_0 \uplus V'_1,E',p')$, such that the winning sets of $G$ can be efficiently recovered from the winning sets of $G'$.
After exhaustive application of the reduction rules, the reduced game~$G'$ will have size bounded by some function of $k$ and $p$ only, independent of the size of~$G$.

The digraphs of our underlying parity game may have self-loops and bidirected edges, but (without loss of generality) no parallel edges between the same two nodes.
Thus, whenever parallel edges arise during the application of one of the reduction rules, we remove one of them without explicit mention.

\begin{lemma}
\label{outdegreecriteria}
  Let $G = (V_0\uplus V_1,E,p)$ be a bipartite parity game, and let $u,v\in V_0$ be such that $N_G^+(v) \subseteq N_G^+(u)$ and $p_1(v) \geq p_1(u)$.
  Let $G'$ be the parity game obtained from $G$ by deleting the edges~$\{(w,u) \in E\mid (w,v)\in E\}$.
  Then the winning sets of $G$ and $G'$ are equal.
\end{lemma}

\begin{proof}
  We show that an edge $(w',u) \in \{(w,u) \in E\mid(w,v) \in E\}$ can only be part of a winning strategy for player~$1$ on node $w'$ if the edge $(w',v)$ is part of a winning strategy for player~$1$ on $w'$ as well.
  Therefore, after deleting $(w',u)$, player~$1$ wins from $w'$ in $G'$ if and only if he wins from $w'$ in $G$.
  Deleting the edges in $\{(w,u) \in E\mid(w,v)\in E\}$ does therefore not change the winning sets.

  Assume that player~$1$ has a winning strategy $s_1: V_1 \rightarrow V_0$ for $w'$ with $s_1(w') = u$.
  Let $s_1': V_1 \rightarrow V_0$ be defined by $s_1'(w') = v$ and $s_1'(w) = s_1(w)$ for all $w \in V_1\setminus\{w'\}$.
  We claim that $s_1'$ is a winning strategy for player~$1$ on $w'$ as well.
  Assume that there exists a counter strategy $s_0'$ for $s_1'$ such that player~$0$ wins the game with starting node $w'$.
  We will define a strategy $s_0$ for player~$0$ and show that $s_0$ is a counter strategy for $s_1$. 
  Note that $s_0$ will not necessarily be a positional strategy.
  For all $w \in V_0\setminus\{u\}$, $s_0$ chooses the same successor as $s_0'$, but on $u$ it might change its behavior.
  Each time the play encounters $u$ directly after encountering~$w'$, strategy $s_0$ chooses $s_0'(v)$ as the successor of $u$.
  Every other time the play encounters~$u$, strategy $s_0$ chooses~$s_0'(u)$ as the successor of $u$.
  
  The play defined by $s_0'$ and $s_1'$ can then be transformed into the play defined by $s_0$ and $s_1$ by replacing every appearance of the sequence $w',v,s_0'(v)$ with the sequence $w',u,s_0'(v)$.
  Let $C'$ be the cycle created by~$s_0'$ and $s_1'$; then~$C'$ is a winning cycle for player~$0$.
  Then $s_0$ and $s_1$ will also create the cycle $C'$, if~$C'$ does not contain the sequence $w',v,s_0'(v)$.
  Let us therefore assume that $C'$ contains the sequence $w',v,s_0'(v)$.
  Let~$C$ be the closed walk obtained, when replacing $v$ with $u$ in $C'$.
  After a finite prefix the play defined by $s_0$ and~$s_1$ will consist of an infinite recurrence of $C$.
  Since we have $p_1(v) \geq p_1(u)$ player~$0$ is wining the play defined by $s_0$ and $s_1$.
  This contradicts that $s_1$ is a winning strategy for player~$1$.
\end{proof}

\begin{lemma}
\label{equalitycriteria}
  Let $G = (V_0\uplus V_1,E,p)$ be a bipartite parity game, and let $u,v\in V_0$ be nodes with~$N_G^+(u) = N_G^+(v)$ and~$p(v) = p(u)$.
  Let $G'$ be the parity game obtained from~$G$ by contracting~$u$ and~$v$ into a new node $v'$ with priority $p(v)$.
  Then~$u$ and~$v$ belong to the same winning set~$\mathsf{win}_i(G)$ in~$G$ and~$v'$ belongs to the winning set~$\mathsf{win}_i(G')$ of the same player in~$G'$.
  For all other nodes the winning sets of~$G$ and~$G'$ coincide.  
\end{lemma}

\begin{proof}
  Note that $u$ and $v$ belong to the winning set of the same player~$i$ in $G$.
  We can assume that player~$0$ chooses the same successor for $u$ and~$v$ in her optimal strategy. 
  Then no cycle created by optimal strategies contains both $u$ and $v$ and,
  after the contraction, each simple cycle that does not contain both~$u$ and $v$ is again a simple cycle with the same priorities.
  Also, each cycle in the contracted game either exists in the original game (i.e., it does not contain~$v'$) or an equivalent cycle, which can be created by replacing $v'$ with~$v$ or $u$, exists in the original game.
  We can also map winning strategies in the original game, where $u$ and $v$ have the same successor into winning strategies in the resulting game and vice versa by simply identifying the successors of $v'$ with the successor of $u$ and $v$ and vice versa, while keeping the rest of the strategy.
  We again assume that in the winning strategies in the original game, $v$ and $u$ have the same successor $w$.
  We then set the successor of $v'$ to $w$ and set the successor of any node $w'$ with successor $v$ or $u$ to $v'$. 
  The other way around $v$ and $u$ get the same successor as $v'$ and any node $w'$ with successor $v'$ gets either $v$ or~$u$ as its successor, depending on which of the edges $(w',v)$ and $(w',u)$ exists in the original game.
  A pair of strategies and the pair of strategies, to which they are mapped to, then create corresponding cycles and must therefore either both be winning for player~$1$ or both be winning for player~0.
\end{proof}

\begin{lemma}
\label{indegreecriteria}
  Let $G = (V_0\uplus V_1,E,p)$ be a bipartite parity game, and let $v \in V(G)$ be such that $N_G^-(v) = \emptyset$.
  Then for the parity game $G' = G - v$ and for $i\in\{0,1\}$, any node $v'\not= v$ is winning for player~$i$ in~$G$ if and only if it is winning for player~$i$ in $G'$.
\end{lemma}

\begin{proof}
  The condition $N_G^-(v) = \emptyset$ implies that $v$ cannot be part of any cycle.
  Let $v \in V_j$; then $v \in \mathsf{win}_j(G)$ is equivalent to the existence of some node $w \in \mathsf{win}_j(G) \cap N_G^+(v)$.
  Since all possible strategies for all nodes except $v$ are also possible strategies in $G-\{v\}$, all nodes in $V\setminus\{v\}$ belong to the same winning set in $G$ and in $G-v$.
  (Notice that $G - \{v\}$ is again a parity game.)
  Once we computed the winning sets for $G-\{v\}$, we can check in time $O(n)$ whether $v \in \mathsf{win}_j(G)$ or $v \in \mathsf{win}_{\neg j}(G)$.
\end{proof}

\begin{lemma}
\label{prioritycriteria}
  Let $G = (V_0\uplus V_1,E,p)$ be a parity game with largest priority $p_{\max} = \max\{p(v)\mid v \in V(G)\}$.
  If $p^{-1}(z) = \emptyset$ for some $z\in\{1,\hdots,p_{\max}\}$ then let $G' = (V_0\uplus V_1,E,p')$ be the parity game obtained from $G$ by setting
    $p'(v) = p(v) - 2$ for all~$v\in V$ with $p(v) > z$ and $p'(v) = p(v)$ for all~$v\in V$ with $p(v) < z$.
  Then the winning sets of the games~$G$ and~$G'$ coincide.
\end{lemma}

\begin{proof}
  Let~$s_0$ and~$s_1$ be strategies for player~$0$ and player~$1$, respectively, and let~$C=(v_0,v_1,\ldots,v_{\ell})$ be the cycle created by these strategies when the
  game starts at some node~$v$. The parity~$j$ of the largest element in the set~$Q=\{p(v_0),\ldots,p(v_{\ell})\}$ determines which player wins in~$G$
  and the parity~$j'$ of the largest element in the set~$Q'=\{p'(v_0),\ldots,p'(v_{\ell})\}$ determines which player wins in~$G'$. It is easy to see that our
  reduction ensures that~$j=j'$. Since this is true for any cycle, the lemma follows.
\end{proof}

\begin{corollary}
  In any parity game with maximum priority $p_{\max}$ to which the reduction rule described in
  Lemma~\ref{prioritycriteria} cannot be applied anymore, the set of priorities is
  either $\{0,1,\hdots,p_{\max}\}$ or $\{1,\hdots,p_{\max}\}$.
\end{corollary}


\begin{lemma}
\label{kernelsize}  
  Let $G=(V_0 \uplus V_1,E,p)$ be a bipartite parity game with $|V_1| = k$ that is reduced according to Lemmas~\ref{outdegreecriteria}--\ref{indegreecriteria}.
  Then $|V_0| \leq 2^k\cdot \min\{k,p\}$.
\end{lemma}

\begin{proof}
  For each node $v \in V_0$ there are $2^k$ possible choices for $N_G^+(v)$.
  Lemma~\ref{outdegreecriteria} yields that for two nodes $v \neq u \in V_0$ with $N_G^+(v) = N_G^+(u)$ we must have $N_G^-(v) \cap N_G^-(u) = \emptyset$.
  Lemma~\ref{indegreecriteria} then yields that there can be at most $k$ nodes in $V_0$ for every possible choice of $N_G^+(v)$.
  Also Lemma~\ref{equalitycriteria} yields that for each possible choice of $(N_G^+(v),p(v))$ there exists at most one node in $V_0$.
\end{proof}

\begin{lemma}
\label{kernelizationtime}
  There exists a sequence of applications of the reduction rules described in Lemmas~\ref{outdegreecriteria}--\ref{prioritycriteria} with a total run time
  of~$O(n^3)$  that leads to a game in which none of these rules applies anymore. 
\end{lemma}

\begin{proof}
  We show for each reduction rule separately how to apply it exhaustively in time $O(n^3)$.
 Although it can happen, that some reductions corresponding to one of the rules could lead to allowing some other reductions which were not allowed before.
 Therefore we cannot only apply all reductions corresponding to one rule after all reductions corresponding to another rule.
 
Most of the run time will be necessary to test if Lemma~\ref{outdegreecriteria} or Lemma~\ref{equalitycriteria} applies to an ordered pair of nodes.
 We will argue how to apply the reductions such that we do not have to test the same ordered pair of nodes more than once, yielding a total run time of $O(n^3)$.

 To apply all reductions of Lemma~\ref{prioritycriteria}, we first sort the nodes in increasing order of their priorities and create an order of subsets each containing all nodes with the same priority; this can be done in $O(n \log(n))$ time.
 We then save for each of the subsets if its corresponding priority is odd or even and unite consecutive sets with the same parity.
 If the parity of the priority in the first subset is even, all nodes in the $i$-th subset get priority $i-1$; otherwise all nodes in the $i$-th subset get priority $i$.
 Uniting sets can be done in linear time, and we cannot unite more than $n$ times.
 The time for applying Lemma~\ref{prioritycriteria} is thus~$O(n^2)$.

 To apply all reductions for Lemma~\ref{outdegreecriteria}, we need to check for each pair of nodes $\{u,v\}\subseteq V_0$ with $p_1(v) \geq p_1(u)$ whether $N_G^+(v) \subseteq N_G^+(u)$, and find all nodes $w$ with $(w,u) \in E$ and $(w,v) \in E$.
 There are $O(n^2)$ node pairs $\{u,v\}\in V_0$ with $p_1(v) \geq p_1(u)$, which can easily be found using the order of subsets created for Lemma~\ref{prioritycriteria}.
 Checking if $N_G^+(v) \subseteq N_G^+(u)$ and finding all nodes $w$ with $(w,u) \in E$ and $(w,v) \in E$ can be done in time~$O(n)$.
 The total run time for Lemma~\ref{outdegreecriteria} therefore is~$O(n^3)$.

 To apply all reductions for Lemma~\ref{equalitycriteria} we need to check for each pair of nodes $\{u,v\}\subseteq V_0$ with $p(v) = p(u)$ whether $N_G^+(v) = N_G^+(u)$.
 There are $O(n^2)$ such pairs $\{u,v\}$ with $p(v) = p(u)$, which can easily be found using the order of subsets created for Lemma~\ref{prioritycriteria}.
 Testing whether $N_G^+(v) = N_G^+(u)$ and identifying~$u$ and $v$ can be done in time $O(n)$.
 The total run time for Lemma~\ref{equalitycriteria} therefore is $O(n^3)$.

 To apply all reductions for Lemma~\ref{indegreecriteria}, we only need to check for each node if it has incoming edges and possibly delete it.
 Testing a node can be done in constant time, and deleting a node takes at most linear time.
 The time for applying Lemma~\ref{indegreecriteria} is thus $O(n^2)$.

 We will first apply all feasible reductions for Lemma~\ref{prioritycriteria}, then all feasible reductions for Lemmas~\ref{outdegreecriteria},~\ref{equalitycriteria} and~\ref{indegreecriteria}.  
 Any reduction that is now possible was not feasible in the beginning. 

 Observe that some reductions can result in other reductions becoming feasible.
 Since we do not change the out-neighborhood of any node in $V_0$, reductions corresponding to Lemmas~\ref{outdegreecriteria} and~\ref{equalitycriteria} for a pair of nodes $\{u,v\}\subseteq V_0$ can only become feasible when we combine the two subsets containing $v$ and $u$ in a reduction corresponding to Lemma~\ref{prioritycriteria}.
 For each node pair $\{u,v\}\subseteq V_0$ this can happen at most once.
 The total run time for all reductions corresponding to Lemma~\ref{outdegreecriteria} and~\ref{equalitycriteria} therefore is in $O(n^3)$.
 Reductions corresponding to Lemma~\ref{indegreecriteria} and a node $v \in V_0$ can only become feasible when we remove incoming edges of $v$.
 This can happen at most~$n$ times for each node $v \in V_0$, before we remove it.
 The total run time for all reductions corresponding to Lemma~\ref{indegreecriteria} therefore is in $O(n^2)$.
 Reductions corresponding to Lemma~\ref{prioritycriteria} can only become feasible when all nodes of one subset have been removed.
 This can happen at most $n$ times; hence any node will be moved to another subset at most $n$ times.
 The total run time for all reductions corresponding to Lemma~\ref{prioritycriteria} therefore is in~$O(n^2)$.
 \end{proof}

We can now prove our main kernelization result.
 
\begin{proof}[Theorem~\ref{thm:kernel-pk-main}]
The part of the theorem for general instances follows directly from Lemma~\ref{lemma:reduction}.
The part for bipartite instances follows from Lemma~\ref{kernelsize} and Lemma~\ref{kernelizationtime}
because the reduced bipartite parity game~$G' = (V'_0 \uplus V'_1,E',p')$ satisfied $|V'_0| \leq 2^k\cdot \min\{k, p\}$ and $|V'_1|\leq k$.
Since $G'$ is bipartite, this implies that it contains at most $k2^k\cdot \min\{k, p\}$ edges.
\end{proof}

\section{A Simple Exponential-Time Algorithm}
\label{sec:exponentialalgorithm}
A simple algorithm with run time~$O(2^n)$ for the solution of parity games originates from the work of McNaughton~\cite{McNaughton1993} and was first presented for parity games by Zielonka~\cite{Zielonka1998}; see also Gr{\"a}del et al.~\cite{GradelEtAl2002}.
Algorithm $\mathbf{win}(G)$ receives a parity game~$G$ and returns the pair of winning sets $(\mathsf{win}_0(G) = W_0, \mathsf{win}_1(G) = W_1)$.

Algorithm $\mathbf{win}(G)$ is based on Lemmas~\ref{thm:winning-set-from-reachability-set} and~\ref{thm:recurse-after-reach-removable}.
Let $p_{\max}$ be the largest priority in~$G$ and let $V_{p_{\max}}$ be the set of nodes with priority $p_{\max}$.
Let $i = p_{\max} \pmod 2$ be the player who owns the highest priority.
The algorithm first finds the winning sets $(W'_0,W'_1)$ of the smaller game $G' = G - \reach_i(V_{p_{\max}})$ in a first recursive call.
If~$W'_{\neg i} = \emptyset$, then by Lemma~\ref{thm:recurse-after-reach-removable} player $i$ wins from all nodes of $G$ and we are done.
Otherwise, again by Lemma~\ref{thm:recurse-after-reach-removable} we know that $W'_{\neg i} \subseteq \mathsf{win}_{\neg i}(G)$.
The algorithm then finds the winning sets $(W''_0 ,W''_1)$ of the smaller game $G'' = G - \reach_{\neg i}(W'_{\neg i})$ by a second recursive call.
By Lemma~\ref{thm:winning-set-from-reachability-set}, $\mathsf{win}_i(G) = W''_i$ and $\mathsf{win}_{\neg i}(G) = \reach_{\neg i}(W'_{\neg i}) \cup W''_{\neg i} = V (G) \setminus W''_i$.

\begin{algorithm*}[htpb]
 \caption{$\mathbf{win}(G)$}
 \label{algorithm:win}
 \begin{algorithmic}[1]
 \Require A parity game $G = (V_0\uplus V_1,E,p)$ with maximum priority $p_{\max}$.
 \Ensure $(W_0,W_1)$, where $W_i$ is the winning set of player $i\in\{0,1\}$.
 \If{$V = \emptyset$}
  \State \textbf{return} $(\emptyset, \emptyset)$
 \EndIf
 \State $i \leftarrow p_{\max} \pmod 2 ; j \leftarrow \neg i$
	\State $(W'_0,W'_1) \leftarrow \textbf{win}(G - \reach_i(V_{p_{\max}}))$
	\If{$W'_j = \emptyset$}
	 \State $(W_i,W_j) \leftarrow (V, \emptyset)$
	\Else 
	 \State $(W''_0,W''_1) \leftarrow \textbf{win}(G - \reach_j(W'_j))$ 
	 \State $(W_i,W_j) \leftarrow (W''_i,V \setminus W''_i)$
	\EndIf\\
 \Return $(W_0,W_1)$
 \end{algorithmic}
\end{algorithm*}

\begin{theorem}
\label{thm:runtimeWin}
 Algorithm $\mathbf{win}(G)$ finds the winning sets of any parity game on~$n$ nodes in time~$O(2^n)$. 
\end{theorem}
\begin{proof}
 The correctness of the algorithm follows from Lemmas~\ref{thm:winning-set-from-reachability-set} and~\ref{thm:recurse-after-reach-removable}, as argued above.
 Let $T'(n)$ be the number of steps needed by algorithm $\mathbf{win}(G)$ to solve a game $G$ on $n$ nodes. 
 Algorithm $\mathbf{win}(G)$ makes two recursive calls $\mathbf{win}(G')$ and $\mathbf{win}(G'')$ on games with at most $n-1$ nodes.
 Other than that, it performs only $O(n^2)$ operations.
 (The most time-consuming operations are the computations of the sets $\reach_i(V_{p_{\max}})$ and $\reach_j(W'_j )$.)
 Therefore, $T'(n) \leq 2T'(n-1) + O(n^2)$, which implies~$T'(n) = O(2^n)$.
\end{proof}

\section{Overview of the New Algorithms}
\label{sec:overviewofthenewalgorithm}

Before we describe our new algorithms that lead to Theorems~\ref{thm:fixed-parameter-k-main}
and~\ref{thm:outdegreebased-algorithm} in detail in Sect.~\ref{sec:newalgorithm} and 
Sect.~\ref{sec:degreebased_algorithm}, we present an overview of the main ideas.
The algorithm~$\textbf{new-win}(G)$ by Jurdzi{\'n}ski, Paterson, and Zwick~\cite{JurdzinskiEtAl2008} with run time~$n^{O(\sqrt{n})}$
is a slight modification of the just described algorithm~$\textbf{win}(G)$.
At the beginning of each recursive call it tests in time~$O(n^{\ell})$ if the parity game contains a dominion $D$ of
size at most~$\ell=\lceil\sqrt{2n}\rceil$. If this is the case then $D$ is removed and the remaining game
is solved recursively. Else, the parity game is solved by the algorithm~$\textbf{win}(G)$,
except that the recursive calls in lines~4 and~8 are made to~$\textbf{new-win}(G)$. Since this happens only when~$G$
does not contain a dominion of size at most~$\ell$, the dominion~$\reach_j(W'_j)$ that is removed in line~8 has
size greater than~$\ell$ and hence, the second recursive call is to a substantially smaller game. Overall, this
leads to the improved run time of~$n^{O(\sqrt{n})}$.      

Our new algorithms are based on a similar idea. Instead of simply searching for a dominion of size at most~$\ell$,
our algorithm~$\textbf{new-win}_1(G)$ that leads to Theorem~\ref{thm:fixed-parameter-k-main} searches for a dominion that contains at most~$\ell = \lfloor \sqrt{2k}\rfloor$
nodes of the odd player, assuming without loss of generality that the odd player controls fewer nodes, i.e., $k=|V_1|$.
If such a dominion is found then we remove it from the game and solve the remaining game recursively. Otherwise, we use the algorithm~$\textbf{win}(G)$
to solve the parity game, except that the recursive calls in lines~4 and~8 are made to~$\textbf{new-win}_1(G)$.
It can happen that in the game to which the first recursive call in line~4 is made, the odd player controls again~$k$ nodes. 
We will show that in bipartite instances this cannot happen in two consecutive calls. For general instances we use
that the observation that at least the number of different priorities decreases by one in the recursive call.
Searching efficiently for a dominion that contains at most~$\ell = \lfloor \sqrt{2k}\rfloor$ nodes of the odd player 
is more involved than simply searching for dominions whose total size is at most~$\ell$. We use multiple recursive calls of~$\textbf{new-win}_1$ to test if such a dominion exists, which makes the recursion of our algorithm and its analysis 
more complicated.  

Our second algorithm leading to Theorem~\ref{thm:outdegreebased-algorithm} is based on the same approach and inspired by the algorithm of Jurdzi{\'n}ski, Paterson, and Zwick~\cite{JurdzinskiEtAl2008}.
In this case we let~$s_j$, for some~$j\in\mathbb{N}$, equal the number of nodes with out-degree at most~$j$. We separate the nodes into~$s_j$ nodes with out-degree at most~$j$ and~$n-s_j$ nodes with out-degree larger than~$j$ and, at the beginning of each iteration, search for and remove dominions that contain at most~$\ell=\lceil\sqrt{2(n-s_j)}\rceil$ nodes with out-degree larger than~$j$ and
at most~$s=\lceil \sqrt{s_j \cdot \log_{j}{s_j}}\rceil$ nodes with out-degree at most~$j$. 
This algorithm runs in time 
$
n^{O\left(\sqrt{n-s_j} + \sqrt{\frac{s_j}{\log_{j}{s_j}}}\right)},
$ 
which implies Theorem~\ref{thm:outdegreebased-algorithm}.\qed 
\section{Finding Small Dominions}
\label{sec:findingsmalldominions}
We now describe how dominions with the previously discussed properties can be found.
Let $G = (V_0\uplus V_1,E,p)$ be a parity game.
Recall that for $i \in \{0,1\}$, a set $D \subseteq V$ is an $i$-dominion if player~$i$ can win from every node of $D$ without ever leaving $D$, regardless of the strategy of player~$\neg i$.
Note that any $i$-dominion must be $i$-closed.
A set $D \subseteq V$ is a \emph{dominion} if it is either a $0$-dominion or a $1$-dominion.
By prefix independence of parity games, the winning set $\mathsf{win}_i(G)$ of player~$i$ is an $i$-dominion.

For $k,p\in\mathbb N$, let $T(k)$ denote the maximum number of steps needed to solve a bipartite parity game $G = (V_0\uplus V_1,E,p)$ and let $T(k,p)$ denote the maximum number of steps needed to solve a general parity game $G = (V_0\uplus V_1,E,p)$ with $|V_1| = k$ and $p = |\{p(v)\mid v \in V\}|$ using some fixed algorithm. 
For $k,p,\ell\in\mathbb N$, let $\mathsf{dom}_k(\ell)$ denote the maximum number of steps required to find a dominion~$D$ with $|V_1 \cap D| \leq \ell$ in a bipartite parity game $G = (V_0\uplus V_1,E,p)$ with $|V_1| = k$ and let $\mathsf{dom}_{k,p}(\ell)$ denote the maximum number of steps required to find a dominion~$D$ with $|V_1 \cap D| \leq \ell$ in a general parity game $G = (V_0\uplus V_1,E,p)$  with $|V_1| = k$ and $p = |\{p(v)\mid v \in V\}|$, or to determine that no such dominion exists.

We will in the analysis of run times make the assumption that computation and removal of reachability sets as well as kernelization are elementary operation and can therefore be performed in time~$O(1)$.
To obtain the actual run times of our algorithms we will in the end multiply the computed run times by a factor corresponding to the time needed for these operations.

\begin{lemma}
\label{dominionsize1}
 For~$k\ge 4$,
   $\mathsf{dom}_k(\ell) = O(k^{\ell} \cdot T(\ell))$ and
   $\mathsf{dom}_{k,p}(\ell) = O(k^{\ell} \cdot T(\ell,p))$.
\end{lemma}

\begin{proof}
  There are~$O(k^{\ell})$ sets~$V_D\subseteq V_1$ with $|V_D| = \ell$. We argue below that for each such set~$V_D$,
  one can determine whether or not there exists a dominion~$D$ with~$D \cap V_1 \subseteq V_D$ by solving two
  parity games that are sub-games of~$G$, i.e., these games arise from~$G$ by removing some of the nodes. This implies
  the lemma because each of these sub-games can be solved in time~$T(\ell)$ or~$T(\ell,p)$ for bipartite or 
  general parity games, respectively.

  Let $V_D \subseteq V_1$ be a set with $|V_D| = \ell$.
  We will now show how to check if there exists an $i$-dominion $D$ with $D \cap V_1 \subseteq V_D$.
  If such an $i$-dominion $D$ exists, then it is $i$-closed.
  Therefore, it does not contain any node $v \in V$ from which player~$\neg i$ can reach a node in $V_1\setminus V_D$.
  Let $V' = V(G) \setminus \reach_{\neg i}(V_1\setminus V_D)$ be the set of nodes from which player~$\neg i$ cannot force to reach a node in $V_1\setminus V_D$; the set $V'$ can therefore be computed by computing and removing a reachability set, which as we assumed is an elementary operations.
  We then have $D \subseteq V'$, and since no node in $V_1\setminus V_D$ can be part of $V'$, it holds that $V' \cap V_1 \subseteq V_D$.
	Since $V'$ is an $i$-closed set, the game $G-\reach_{\neg i}(V_1\setminus V_D)$ is well defined.
  Let $\mathsf{win}_{i}(V')$ be the winning set of player~$1$ in the game~$G-\reach_{\neg i}(V_1\setminus V_D)$.
  Then~$\mathsf{win}_{i}(V')$ is an $i$-dominion that contains $D$.

  This shows that for each set $V_D \subseteq V_1$ with $|V_D| = \ell$ we only need to compute for $i \in \{0,1\}$ the sets $V_i' = V \setminus \reach_{\neg i}(V_1\setminus V_D)$ of nodes from which player~$\neg i$ cannot force to enter $V_1 \setminus V_D$ and compute the winning sets 
  of the game~$G-\reach_{\neg i}(V_1\setminus V_D)$ to determine whether or not there exists a dominion~$D$ with~$D \cap V_1 \subseteq V_D$. 
\end{proof}

With the algorithm described in Lemma~\ref{dominionsize1} we can find a dominion $D$ such that $|D \cap V_1| \leq \ell$ if such a dominion exists.
We denote this algorithm by~$\textbf{dominion}_1(G,\ell)$ and assume that it returns either the pair $(D, i)$ if an $i$-dominion~$D$ is found, or $(\emptyset,-1)$ if not.

We will give the pseudocode for algorithm $\textbf{dominion}(G,\ell,s)$.
In the pseudocode, let $U_m$ denote the set of marked nodes from $U$ and let $\textbf{king}(U,\text{strategy}_i)$ denote an execution of the algorithm by King et al.~\cite{King2001} that determines the winners of the sub-game~$G$ restricted to~$U$ with a given strategy for player~$i$.
\begin{algorithm*}[ht]
 \caption{$\textbf{dominion}(G,\ell,s)$}
 \label{algorithm:dominion}
 \begin{algorithmic}[1]
  \Require A parity game $G = (V_0\uplus V_1,E,p)$ and $\ell, s \in \{ 0,\ldots, |V(G)|\}$.
  \Ensure An $i$-dominion $(D,i)$ for $i \in \{0,1\}$ or $(\emptyset,-1)$ if no dominion is found.
  \State Fix a total order $\prec$ on the nodes of $G$.
	\State For each $u \in V(G)$ sort the edges emanating from $u$ by $\prec$ on their respective endpoint.
	\For{$i \in \{0,1\}$}
		\For{$v \in V_i$}
			\For{$\left\langle a_1,\ldots,a_{\ell},b_1,\ldots,b_s\right\rangle \in \{1,\ldots,|V(G)|\}^{\ell} \times \{1,\ldots,j\}^s$}
				\State $r_1 = 1, r_2 = 1, U = \{v\}, U_m = \emptyset$
				\While{$U \neq \emptyset$, $r_1 \leq \ell$ and $r_2 \leq s$}
					\State Choose $u = \min(U,\prec)$.
					\State $U = U\setminus \{u\}$, $U_m = U_m \cup \{u\}$.
					\If{$u \in V_i$}
						\If{$|\delta^+(u)| > j$}
							\If{$|\delta^+(u)| \leq r_1$}
								\State Let $e = (u,w)$ be the $a_{r_1}$-th outgoing edge of $u$.
								\State $U = U \cup (\{w\}\setminus U_m)$, $r_1 = r_1 +1$, $\text{strategy}_i(u) = w$.
							\Else
								\State $U = \emptyset$, $U_m = \emptyset$.
							\EndIf
						\Else
							\If{$|\delta^+(u)| \leq r_2$}
								\State Let $e = (u,w)$ be the $b_{r_2}$-th outgoing edge of $u$.
								\State $U = U \cup (\{w\}\setminus U_m)$, $r_2 = r_2 +1$, $\text{strategy}_i(u) = w$.
							\Else
								\State $U = \emptyset$, $U_m = \emptyset$.
							\EndIf
						\EndIf
					\Else
						\State $U = U \cup (N^+(u)\setminus U_m)$
					\EndIf
				\EndWhile
				\If{$U_m \neq \emptyset$ contains at most $\ell$ high out-degree nodes and at most $s$ low out-degree nodes}
					\State $(W_0,W_1) = \textbf{king}(U,\text{strategy}_i)$.
					\If{$W_i = U$}
						\State \Return{$(U,i)$}
					\EndIf
				\EndIf
			\EndFor
		\EndFor
	\EndFor
  \State \Return $(\emptyset,-1)$.
 \end{algorithmic}
\end{algorithm*}

\section{New Algorithms for Solving Parity Games}
\label{sec:newalgorithm}

We present the algorithm~$\textbf{new-win}_1(G)$ discussed in Sect.~\ref{sec:overviewofthenewalgorithm} in detail. 
Let~$G = (V_0\uplus V_1,E,p)$ with $|V_1| = k$ be a parity game with~$p$ distinct priorities.

The algorithm~$\textbf{new-win}_1$ starts by trying to find a ``small'' dominion $D$, where small means~$|D \cap V_1| \leq \ell$, where $\ell = \lfloor \sqrt{2k}\rfloor$ is a parameter chosen to minimize the run time of the algorithm.
If such an $i$-dominion is found, then we remove it together with its $i$-reachability set from the game and solve the remaining game recursively.
If no small dominion is found, then~$\textbf{new-win}_1$ simply calls algorithm~$\textbf{old-win}_1$, which is almost identical to algorithm $\textbf{win}$.
The only difference between $\textbf{old-win}_1$ and $\textbf{win}$ is that its recursive calls are made to $\textbf{new-win}_1$ and not to itself.

The recursion stops once the number of odd nodes is at most $4$, in which case we will test each of the at most $((p+1)^4)^4$ (due to the size of our kernel) different strategies for player~$1$ in constant time.
We will call this brute force method $\textbf{solve}(G)$.
We will also kernelize using the reduction rules described in Sect.~\ref{sec:reductionofparitygames}. We will call the kernelization subroutine $\textbf{kernel}(G)$.
The pseudocode of $\textbf{new-win}_1(G)$ can be found in Sect.~\ref{sec:pseudocode-new-win}.

The correctness of the algorithm follows analogously to the correctness of~$\textbf{win}(G)$.
We analyze the run time of~$\textbf{new-win}_1(G)$
and prove Theorem~\ref{thm:fixed-parameter-k-main} in section~\ref{sec:RunTimeAppendix}.

\section{Out-degree based Algorithm}
\label{sec:degreebased_algorithm}
We now describe our second algorithm~$\textbf{new-win}_2(G,j)$.
In order to describe it, let~$j\in\mathbb{N}$ and let~$s_j$ denote the number of nodes of out-degree at most~$j$. 
$\textbf{new-win}_2(G,j)$ is then almost identical to $\textbf{new-win}_1(G)$, but instead of dominions that contains at most~$\ell' = \lfloor \sqrt{2k}\rfloor$
nodes of the odd player, we search for and delete dominions that contain at most~$\ell=\lceil\sqrt{2(n-s_j)}\rceil$ nodes with out-degree larger than~$j$ and
at most~$s=\lceil \sqrt{s_j \cdot \log_{j}{s_j}}\rceil$ nodes with out-degree at most~$j$. 
This algorithm has a run time of 
$
n^{O\left(\sqrt{n-s_j} + \sqrt{\frac{s_j}{\log_{j}{s_j}}}\right)},
$ 
which implies Theorem~\ref{thm:outdegreebased-algorithm}.

In the following let us assume~$j =\arg\min_{1\leq j' \leq n} \Big\{\sqrt{n-s_{j'}}+\sqrt{\frac{s_{j'}}{\log_{j'}{s_{j'}}}}\Big\}$. 
We say that a node~$v$ has \emph{high out-degree} if $|\delta^+(v)| > j$ and \emph{low out-degree} otherwise.
For $n,z,\ell,s\in\mathbb N$, let $\mathsf{dom}_{n,z}(\ell,s)$ denote the maximum number of steps required to find a dominion~$D$ with at most $\ell$ high out-degree and at most $s$ low out-degree nodes in a parity game $G = (V_0\uplus V_1,E,p)$ with $n$ nodes out 
of which $z$ are high out-degree nodes, or to determine that no such dominion exists.

\begin{lemma}
\label{dominionsize2}
 For all values of $\ell,s \in \{0,\ldots,n\}$, it holds
 \[
    \mathsf{dom}_{n,z}(\ell,s) = O\left(n^{\ell+1} j^s (\ell +s)^2 \cdot \max\{1,\log(\ell+s)\}\right)
                               = O\left(n^{\ell+4} j^s\right).
 \]
\end{lemma}

\begin{proof}
  Fix an arbitrary total order $\prec$ on $V(G)$.
  Let $u \in V(G)$ be a node of $G$ and let $(u, v_1), \ldots, (u,v_{|\delta^+(u)|})$ be the edges   emanating from $u$, where $v_i \prec v_{i+1}$ for all $i \in \{1,\ldots|\delta^+(v)|-1\}$; we call $(u, v_i)$ the \emph{$i$-th outgoing edge of $u$}.
  The algorithm generates at most $O(n\cdot n^{\ell} j^s)$ $0$-closed sets of nodes that contain at most $\ell$ nodes with an out-degree greater than~$j$ and at most $s$ nodes with an out-degree at most $j$, which are candidates for being $0$-dominions.
  For every node $v \in V$ and every sequence $\left\langle a_1, \ldots , a_{\ell}, b_1, \ldots,b_s\right\rangle \in \{1,\ldots, n \}^{\ell}\times\{1,\ldots,j\}^{s}$ construct a set $U \subseteq V$ as follows.
  Start with $U = \{v\}$ and $r_1 = 1, r_2 = 1$.
  Nodes added to $U$ are initially unmarked.
  As long as there is still an unmarked node in $U$, pick the smallest such node $u \in U$ with respect to $\prec$ and mark it.
  \begin{itemize}
    \item If $u \in V_0$ and $u$ has high out-degree then add the endpoint of the $a_{r_1}$-th outgoing edge of $u$ to $U$ (if it is not already present in $U$) and increment $r_1$.
    \item If $u \in V_0$ and $u$ has low out-degree then add the endpoint of the $b_{r_2}$-th outgoing edge of $u$ to $U$ (if it is not already present in $U$) and increment $r_2$.
    \item If $u \in V_1$ then add the endpoints of all outgoing edges of $u$ that are not yet part of $U$ to $U$.
  \end{itemize}
  If at some stage $U$ contains either more than $\ell$ nodes with high out-degree or more than $s$ nodes with low out-degree, or the endpoint of the $i$-th outgoing edge of some node $v$ with out-degree $|\delta^+(v)| < i$ should be added to $U$, then discard the set $U$ and restart the construction
with the next sequence.
  If the process above ends without discarding $U$, then a $0$-closed set containing at most $\ell$ high out-degree and at most $s$ low out-degree nodes has been found.
  Furthermore, for every node $u \in U \cap V_0$, one of the outgoing edges of~$u$ was selected.
  This corresponds to a suggested strategy for player~$0$ in the game $G$ restricted to the set~$U$.
  
  Our algorithm therefore considers by exhaustive search all $0$-closed sets containing at most $\ell$ high out-degree and at most $s$ low out-degree nodes, 
and for each set considers all possible positional strategies for player~$0$.
  Using an algorithm of King et al.~\cite{King2001} we can check in time $O((\ell +s)^2 \log(\ell+s) )$ time whether a given pair of set $U$ and proposed strategy is indeed a winning strategy for
player~$0$ from all nodes of $U$. 
  Thus, if there is a $0$-dominion containing at most~$\ell$ high out-degree and at most $s$ low out-degree nodes, then the algorithm will find one.
  Finding $1$-dominions can be done in an analogous manner.
\end{proof}

With the just described algorithm, we can find a dominion~$D$ with at most $\ell$ nodes with high out-degree and at most $s$ nodes with low out-degree if such a dominion exists.
We denote this algorithm by $\textbf{dominion}_2(G,\ell,s)$, and suppose that it returns either the pair $(D,i)$ if such an $i$-dominion $D$ is found, or $(\emptyset,-1)$ if not.

The algorithm~$\textbf{new-win}_2$ starts by trying to find a ``small'' dominion $D$, where small means that~$D$ contains 
at most $\ell$ nodes with out-degree greater than~$j$ and at most $s$ nodes with out-degree at most $j$, where $\ell = \lceil \sqrt{2(n-s_j)}\rceil$ and $s = \lceil \sqrt{s_j \cdot \log_{j}{s_j}}\rceil$ are parameters chosen to minimize the run time of the whole algorithm.
If such an $i$-dominion is found, then we remove it together with its $i$-reachability set from the game and solve the remaining game recursively.
If no small dominion is found, then $\textbf{new-win}_2$ simply calls algorithm $\textbf{old-win}_2$, which is almost identical to algorithm $\textbf{win}$.
The only difference between $\textbf{old-win}_2$ and $\textbf{win}$ is that its recursive calls are made to $\textbf{new-win}_2$ and not to itself.

The recursion stops once the number of nodes with out-degree at most~$j$ and the number of nodes with out-degree greater than $j$ are both at most $3$, in which case we will test all of the at most constant different strategies for the two players in constant time.
We will call this brute force method \textbf{solve($G$)}.
The pseudocode of $\textbf{new-win}_2$ can be found in Sect.~\ref{sec:pseudocode-new-win}.

The correctness of $\textbf{new-win}_2$ follows analogously to the correctness of the simple algorithm~$\textbf{win}$.
We analyze the run time of~$\textbf{new-win}_2$
and prove Theorem~\ref{thm:outdegreebased-algorithm} in Sect.~\ref{sec:RunTimeAppendix}.

We can now proove Corollary~\ref{cor:outdegreebased-algorithm}.

\begin{proof}[Corollary~\ref{cor:outdegreebased-algorithm}]
  First consider a parity game on $n$ nodes played on a graph with maximum out-degree~$\Delta$.
  Then~$s_{\Delta}=n$ and
  \begin{equation*}
   \sqrt{n-s_{\Delta}} + \sqrt{\frac{s_{\Delta}}{\log_{\Delta}{s_{\Delta}}}}
   = \sqrt{\frac{n}{\log_{\Delta}{n}}}
   = \sqrt{\frac{\log(\Delta)n}{\log{n}}} \enspace .
  \end{equation*}
Now the first part of the corollary follows immediately from Theorem~\ref{thm:outdegreebased-algorithm}.

  Let us now consider the case that the average out-degree is~$\Delta$ and let~$z=\log(n)\Delta$.
  Then Markov's inequality implies~$s_{z} \ge (1-1/\log(n))n$. Hence,
    \begin{equation*}
       \sqrt{n-s_{z}} + \sqrt{\frac{s_{z}}{\log_{z}{s_{z}}}}
        \le \sqrt{\frac{n}{\log(n)}} + \sqrt{\frac{n}{\log_{z}{n}}}
      = \sqrt{\frac{n}{\log(n)}} + \sqrt{\frac{\log(\log(n)\Delta)n}{\log{n}}} \enspace .
    \end{equation*}
    Now the second part of the corollary follows immediately from Theorem~\ref{thm:outdegreebased-algorithm}.
\end{proof}

\section{Pseudocode for Algorithms $\textbf{new-win}$}
\label{sec:pseudocode-new-win}
We will now give the pseudocode for the algorithms $\textbf{new-win}_1(G)$ and $\textbf{new-win}_2(G,j)$ together with their subroutines $\textbf{old-win}_1(G)$ and $\textbf{old-win}_2(G,j)$.
In the pseudocodes, we call a function $\textbf{solve}(G)$. This function denotes a bruteforce method to solve parity games and is only used on very small games.

\begin{algorithm}[H]
 \caption{$\textbf{new-win}_1(G)$}
 \label{algorithm:new-win}
 \begin{algorithmic}[1]
   \Require A parity game $G = (V_0\uplus V_1,E,p)$.
   \Ensure A partition $(W_0,W_1)$ of $V$, where $W_i$ is the winning set of player~$i\in\{0,1\}$.
   \State $k \leftarrow |V_1| ; \ell \leftarrow \left\lfloor \sqrt{2k}\right\rfloor$ ; $G = \textbf{kernel}(G)$
	\If{$k\leq 4$}
	 \Return $\textbf{solve}(G)$
   \EndIf
   \State $(D,i) \leftarrow \textbf{dominion}_1(G,\ell)$		
   \If{$D = \emptyset$}
    \State $(W_0,W_1) \leftarrow \textbf{old-win}_1(G)$
   \Else
    \State $(W'_0,W'_1) \leftarrow \textbf{new-win}_1(G - \reach_i(D))$ 
    \State $(W_{\neg i},W_i) \leftarrow (W'_{\neg i},V \setminus W'_{\neg i})$
   \EndIf\\
   \Return $(W_0,W_1)$
 \end{algorithmic}
\end{algorithm}

\begin{algorithm}[H]
 \caption{$\textbf{old-win}_1(G)$}
 \label{algorithm:old-win}
 \begin{algorithmic}[1]
   \Require A parity game $G = (V_0\uplus V_1,E,p)$.
   \Ensure A partition $(W_0,W_1)$ of $V$, where $W_i$ is the winning set of player~$i\in\{0,1\}$.
	 \State $G = \textbf{kernel}(G)$
   \State $i \leftarrow p_{\max} \pmod 2$
   \State $(W'_0,W'_1) \leftarrow \textbf{new-win}_1(G - \reach_i(V_{p_{\max}}))$
   \If{$W'_{\neg i} = \emptyset$}
    \State $(W_i,W_{\neg i}) \leftarrow (V, \emptyset)$
   \Else
    \State $(W''_0,W''_1) \leftarrow \textbf{new-win}_1(G - \reach_{\neg i}(W'_{\neg i}))$ 
	 \State $(W_i,W_{\neg i}) \leftarrow (W''_i,V \setminus W''_i)$
   \EndIf\\
   \Return $(W_0,W_1)$
 \end{algorithmic}
\end{algorithm}

\begin{algorithm}[H]
 \caption{$\textbf{new-win}_2(G,j)$}
 \label{algorithm:new-win2}
 \begin{algorithmic}[1]
   \Require A parity game $G = (V_0\uplus V_1,E,p)$ and $j \in \{1,\ldots |V|\}$.
   \Ensure A partition $(W_0,W_1)$ of $V$, where $W_i$ is the winning set of player~$i\in\{0,1\}$.
	\State $s_j \leftarrow |\{v \in V | |\delta^+(v)| \leq j\}| ; \ell \leftarrow \left\lceil \sqrt{2(n-s_j)}\right\rceil ; s \leftarrow \left\lceil \sqrt{s_j\cdot\log_j{s_j}}\right\rceil$
	\If{$s_j\leq 3$ and $n-s_j \leq 3$}
	 \Return $\textbf{solve}(G)$
   \EndIf
   \State $(D,i) \leftarrow \textbf{dominion}_2(G,\ell,s)$		
   \If{$D = \emptyset$}
    \State $(W_0,W_1) \leftarrow \textbf{old-win}_2(G,j)$
   \Else
    \State $(W'_0,W'_1) \leftarrow \textbf{new-win}_2(G - \reach_i(D),j)$ 
    \State $(W_{\neg i},W_i) \leftarrow (W'_{\neg i},V \setminus W'_{\neg i})$
   \EndIf\\
   \Return $(W_0,W_1)$
 \end{algorithmic}
\end{algorithm}

\begin{algorithm}[H]
 \caption{$\textbf{old-win}_2(G,j)$}
 \label{algorithm:old-win2}
 \begin{algorithmic}[1]
   \Require A parity game $G = (V_0\uplus V_1,E,p)$.
   \Ensure A partition $(W_0,W_1)$ of $V$, where $W_i$ is the winning set of player~$i\in\{0,1\}$.
   \State $i \leftarrow p_{\max} \pmod 2$
   \State $(W'_0,W'_1) \leftarrow \textbf{new-win}_2(G - \reach_i(V_{p_{\max}}),j)$
   \If{$W'_{\neg i} = \emptyset$}
    \State $(W_i,W_{\neg i}) \leftarrow (V, \emptyset)$
   \Else
    \State $(W''_0,W''_1) \leftarrow \textbf{new-win}_2(G - \reach_{\neg i}(W'_{\neg i}),j)$ 
	 \State $(W_i,W_{\neg i}) \leftarrow (W''_i,V \setminus W''_i)$
   \EndIf\\
   \Return $(W_0,W_1)$
 \end{algorithmic}
\end{algorithm}

\section{Analysis of the Run Time}\label{sec:RunTimeAppendix}

We will show that algorithm $\textbf{new-win}(G)$ has a run time of $O(p \cdot m\cdot n) \cdot (p+k)^{O(\sqrt{k})}$ on general instances and in time $O(n^3) \cdot k^{O(\sqrt{k})}$ on bipartite instances. We will also show that algorithm $\textbf{new-win}(G,j)$ has a run time of $n^{O\left(\sqrt{n-s_j} + \sqrt{\frac{s_j}{\log_{j}{s_j}}}\right)}$, where $s_j$ the number of nodes in $G$ with out-degree at most $j$.

Note that the part $O(p m n)$ of the run time comes from the reduction of the instance and the computation and removal of reachability sets of found dominions. Since we do both of these often,
we assume them to be elementary computations with computation time $O(1)$ and show that the total run time remaining is $(p+k)^{O(\sqrt{k})}$.
In bipartite instances we need $O(n^3)$ time to reduce the instance and to compute and remove reachability sets.
We will show that the run time on bipartite instances is $k^{O(\sqrt{k})}$, when computation and removal of reachability sets and the reductions are viewed as an elementary operation.
Let $T(k,p)$ denote the time required by algorithm \textbf{new-win} on a game $G = (V_0\uplus V_1,E,p)$ with $|V_1| = k$ and $p = |\{p(v)\mid v \in V\}|$, when reduction of an instance and computation and removal of reachability sets of found dominions are viewed as elementary computations and have run time $O(1)$. 
\begin{lemma}
\label{lemma:recurrences1}
  The following recurrence relation holds:
  \begin{equation*}
	 (a)  \quad T(k,p) \leq \max\{T(k-1,p), T(k,p-1) + T(k-\ell,p)\}
						       + \mathsf{dom}_{k,p}(\ell) + O(1) \enspace .
 \end{equation*}
\end{lemma}
\begin{proof}
  Algorithm $\textbf{new-win}(G)$ tries to find dominions $D$ with $|D\cap V_1| \leq\ell= \lfloor \sqrt{2k}\rfloor$.
  By definition this takes at most $\mathsf{dom}_{k,p}(\ell)$ time on general instances. If a (non-empty) dominion is found, then the algorithm simply proceeds on the remaining game, which has at most $k-1$ odd nodes, and thus it solves this game in time bounded by $T(k-1,p)$.
  Otherwise, a call to $\textbf{old-win}(G)$ is made.
  This results in a call to $\textbf{new-win}(G - \reach_i(V_{p_{\max}}))$.
  In this case the call takes at most $T(k,p-1)$ time because we removed all nodes with the highest priority.
  If the set $W'_j$ returned by the call is empty, then we are done.
  Otherwise, $W'_j = \mathsf{win}_j(G - \reach_i(V_{p_{\max}}))$ and $W'_j\subseteq \mathsf{win}_j(G)$ by Lemma~\ref{thm:winning-sets-closed-under-subgames}.
  Therefore, $W'_j$ is a $j$-dominion of $G$.
  We are in the case that there is no dominion $D$ with $|D\cap V_1|\le \ell$ in~$G$.
  Thus, $|W'_j \cap V_1| > \ell$, and hence the second recursive call $\textbf{new-win}(G - \reach_j(W'_j))$ takes time at most $T(k-\left\lceil \ell\right\rceil,p)$.
  Consequently,
  \begin{equation*} 
    T(k,p) \leq \max\{T(k-1,p), T(k,p-1) + T(k-\left\lceil \ell\right\rceil,p)\} + \mathsf{dom}_{k,p}(\ell) + O(1).\quad\qedhere
  \end{equation*}
\end{proof}

Let $T(k,\mathsf{even})$ and $T(k,\mathsf{odd})$ denote the time required by algorithm \textbf{new-win} on a bipartite game $G = (V_0\uplus V_1,E,p)$ with $|V_1| = k$ when the largest priority is even respectively odd, when computation and removal of reachability sets of found dominions and the reductions are viewed as elementary computations.
We denote by $T(k)$ the time required by algorithm \textbf{new-win} on any bipartite game with $|V_1| = k$; thus $T(k) \leq \max\{T(k,\mathsf{even}),T(k,\mathsf{odd})\}$.
\begin{lemma}
\label{lemma:recurrences2}
  The following recurrence relations hold:
  \begin{alignat*}{2}
   (b_1)\quad && T(k,\mathsf{odd}) & \leq T(k-1) + T(k-\ell) + \mathsf{dom}_k(\ell) +O(1),\\
   (b_2)\quad && T(k,\mathsf{even}) & \leq \max\{T(k,\mathsf{odd}), T(k-1)\} + T(k-\ell) + \mathsf{dom}_k(\ell)+O(1)\\
              &&		& \leq T(k-1) + 2T(k-\ell) + 2\mathsf{dom}_k(\ell)+O(1),\\
   (b_3)\quad && T(k) & \leq T(k-1) + 2T(k-\ell) + 2\mathsf{dom}_k(\ell)+O(1) \enspace .
 \end{alignat*}
\end{lemma}
\begin{proof}
  From the definition it follows directly that $T(k) \leq \max\{T(k,\mathsf{even}),T(k,\mathsf{odd})\}$.
  Showing $(b_1)$ and $(b_2)$ therefore yields $(b_3)$.
  Algorithm $\textbf{new-win}(G)$ tries to find dominions $D$ with $|D\cap V_1| \leq\ell= \lceil \sqrt{2k}\rceil$.
  By definition this takes at most $\mathsf{dom}_k(\ell)$ time on bipartite instances.
  If a (non-empty) dominion is found, then the algorithm simply proceeds on the remaining game, which has at most $k-1$ odd nodes, and the remaining run time is therefore at most $T(k-1)$.
  Otherwise, a call to $\textbf{old-win}(G)$ is made.
  This results in a call to $\textbf{new-win}(G - \reach_i(V_{p_{\max}}))$.
  Here we have to distinguish whether the highest priority is odd or even.

  If the highest priority is odd then, by Lemma~\ref{indegreecriteria}, the set~$\reach_1(V_{p_{\max}}) \cap V_1$ is non-empty and the call takes at most $T(k - 1)$ time.

 In case the highest priority is even, we either have $\reach_0(V_{p_{\max}}) \cap V_1 \neq \emptyset$ or $\reach_0(V_{p_{\max}}) = V_{p_{\max}}$ in which case Lemma~\ref{prioritycriteria} yields that in $G - \reach_i(V_{p_{\max}})$ the highest priority has to be odd.
 Therefore, this call needs time at most $\max\{T(k,\mathsf{odd}), T(k-1)\}$.

  If the set $W'_j$ returned by the call is empty, then we are done.
  Otherwise, $W'_j = \mathsf{win}_j(G - \reach_i(V_{p_{\max}}))$ and this is part of $\mathsf{win}_j(G)$ by Lemma~\ref{thm:winning-sets-closed-under-subgames}.
  Therefore,~$W'_j$ is a $j$-dominion of $G$.
  We are in the case that there is no dominion $D$ with $|D\cap V_1|$ at most $\ell$ in $G$, so we
know that $|W'_j \cap V_1| > \ell$ , and therefore the second recursive call $\textbf{new-win}(G - \reach_j(W'_j))$
takes at most $T(k-\ell)$ time.
  Thus, we obtain 
  \begin{equation*}
   T(k,\mathsf{odd}) \leq T(k-1) + T(k-\ell) + \mathsf{dom}_k(\ell)
  \end{equation*}
  and
  \begin{align*}
   T(k,\mathsf{even}) & \leq \max\{T(k,\mathsf{odd}),T(k-1)\} + T(k-\ell) + \mathsf{dom}_k(\ell)\\
                & \leq T(k-1) + 2T(k-\ell) + 2\mathsf{dom}_k(\ell),
  \end{align*} 
  which yields $T(k) \leq T(k-1) + 2T(k-\ell) + 2\mathsf{dom}_k(\ell)$.
\end{proof}

For $j\in\mathbb N_0$, let $T'(s_j,n-s_j)$ denote the time required by algorithm \textbf{new-win} on a game $G$ on $n$ nodes of which $s_j$ nodes have out-degree at most $j$.

\begin{lemma}
\label{lemma:recurrences3}
  The following recurrence relation holds:
  \begin{alignat*}{2}
	 (c)  \quad && T'(s_j,n-s_j) & \leq \max\{T'(s_j-1,n-s_j),T'(s_j,n-s_j-1)\}\\
							&&		& \quad + \max\{T'(s_j-s,n-s_j),T'(s_j,n-s_j-\ell)\}
					  + \mathsf{dom}_{n,n-s_j}(\ell,s) +O(1) \enspace .
 \end{alignat*}
\end{lemma}
\begin{proof}
  Algorithm $\textbf{new-win}(G,j)$ tries to find dominions $D$ containing at most~$s$ nodes with out-degree at most $j$ and at most $\ell$ nodes with out-degree greater than $j$.
  By definition this takes at most $\mathsf{dom}_{n,n-s_j}(\ell,s)$ time.
  If a (non-empty) dominion is found, then the algorithm simply proceeds on the remaining game, which has at most $n-1$ nodes, and the remaining time is therefore at most $\max\{T'(s_j-1,n-s_j),T'(s_j,n-s_j-1)\}$.
  Otherwise, a call to $\textbf{old-win}(G,j)$ is made.
  This results in a call to $\textbf{new-win}(G - \reach_i(V_{p_{\max}}),j)$, this call is to a game with fewer nodes and can be solved in time bounded by
  \begin{equation*}  
    \max\{T'(s_j-1,n-s_j),T'(s_j,n-s_j-1)\} \enspace .
  \end{equation*}
  If the set $W'_k$ returned by the call is empty, then we are done.
  Otherwise, $W'_k = \mathsf{win}_k(G - \reach_i(V_{p_{\max}}))$, and $W'_k\subseteq\mathsf{win}_k(G)$ by Lemma~\ref{thm:winning-sets-closed-under-subgames}.
  Therefore, $W'_k$ is a $k$-dominion of $G$.
  We are in the case that there is no dominion $D$ containing at most $s$ nodes with out-degree at most $j$ and at most $\ell$ nodes with out-degree greater than $j$, so $W'_k$ either contains more than $s$ nodes with out-degree at most $j$ or more than $\ell$ nodes with out-degree greater than $j$, and therefore the second recursive call $\textbf{new-win}(G - \reach_k(W'_k))$ takes time bounded by $\max\{T'(s_j-s,n-s_j),T'(s_j,n-s_j-\ell)\}$.

  All other computations can be done in constant time.
  Thus, we obtain 
  \begin{align*}
    T'(s_j,n-s_j) & \leq \max\{T'(s_j-1,n-s_j),T'(s_j,n-s_j-1)\} \\
							& + \max\{T'(s_j-s,n-s_j),T'(s_j,n-s_j-\ell)\}
						 + \mathsf{dom}_{n,n-s_j}(\ell,s) +O(1)\enspace .\quad\qedhere
  \end{align*}
\end{proof}

We analyze recurrences $(a)$ and $(b)$ with $\ell = \lfloor \sqrt{2k}\rfloor$ in Theorem~\ref{thm:solve-recurrences1} in Sect.~\ref{sec:recurrencerelationcomputations}, which eventually shows that $T(k,p) \leq (p+k)^{O(\sqrt{k})}$ and $T(k) \leq k^{O(\sqrt{k})}$, and recurrence $(c)$ with $\ell = \left\lceil  \sqrt{2(n-s_j)}\right\rceil$ and  $s = \left\lceil \sqrt{\frac{s_j}{\log_j{s_j}}}\right\rceil $ in Theorem~\ref{thm:solve-recurrences2} in Sect.~\ref{sec:recurrencerelationcomputations}, which eventually shows that $T'(s_j,n-s_j) \leq n^{O\left(\sqrt{n-s_j}+\sqrt{\frac{s_j}{\log_j{s_j}}}\right)}$.
This completes the analysis of the run time of $\textbf{new-win}(G)$ and $\textbf{new-win}(G,j)$, and it proves Theorem~\ref{thm:fixed-parameter-k-main} and Theorem~\ref{thm:outdegreebased-algorithm}.\qed

\section{Recurrence Relation Computations}
\label{sec:recurrencerelationcomputations}
In this section we analyze the recurrence relations used to bound the run time of~\textbf{new-win}.
\begin{theorem}
\label{thm:solve-recurrences1}
 For $k\in\mathbb N$ and $\ell = \lfloor\sqrt{2k}\rfloor$, we obtain
 \begin{eqnarray*}
	 T(k,p) & = & (p+k)^{O(\sqrt{k})},\\ 
   T(k) & = & k^{O(\sqrt{k})}.
 \end{eqnarray*}
\end{theorem}
To prove Theorem~\ref{thm:solve-recurrences1}, we first establish some lemmas.
\begin{lemma}
 For $k,p\in\mathbb N$ and $\ell = \lfloor\sqrt{2k}\rfloor$, it holds
 \begin{eqnarray*}
  T(k,p) & \leq & 2(k+p)^{\left\lfloor \sqrt{2k}\right\rfloor} \cdot \mathsf{dom}_{k,p}(\left\lfloor \sqrt{2k}\right\rfloor) \enspace \text{ and}\\
   T(k) & \leq & 2(2k){\left\lfloor \sqrt{2k}\right\rfloor} \cdot \mathsf{dom}_k(\left\lceil \sqrt{2k}\right\rceil) \enspace .
 \end{eqnarray*}
\end{lemma}
\begin{proof}
 For every pair of integers $k$ and $p$ we construct binary trees $T_{k,p}$ and $T_{k}$ in the following way.
 The root of $T_{k,p}$ is labeled by $k$ and $k+p$ and the root of~$T_k$ is labeled by $k$.
 A node labeled by a number $k > 4$ has two children: in~$T_{k,p}$ a left child labeled by $k$ and $k+p-1$ and a right child labeled by $k - \left\lceil \sqrt{2k}\right\rceil $ and $p+k - \left\lceil \sqrt{2k}\right\rceil$. In $T_{k}$ a left child labeled by $k'$ and a right child labeled by $k - \left\lceil \sqrt{2k}\right\rceil$.
 A node labeled by $k'$ in $T_{k}$ has two children: a left child labeled by $k-1$ and a right child labeled by $k - \left\lceil \sqrt{2k}\right\rceil$.
 Nodes labeled by a number $k\leq 4$ are leaves.
 A node labeled by $k$ and $k+p$ has a cost of $\mathsf{dom}_{k,p}(\left\lfloor \sqrt{2k}\right\rfloor)$ associated with it and a node labeled by $k$ or $k'$ has a cost of $\mathsf{dom}_k(\left\lceil \sqrt{2k}\right\rceil))$ associated with it.
 It follows from Lemma~\ref{lemma:recurrences1} and Lemma~\ref{lemma:recurrences2} that the sum of the costs of the nodes of $T_{k,p}$ and $T_k$,  is an upper bound on $T(k,p)$ and $T(k)$, respectively.
 Clearly, the length of every path from the root to a leaf is at most $p+k+1$ in~$T_{k,p}$ and $2k$ in~$T_{k}$.
 We say that such a path makes a right turn when it descends from a node to its right child.
 We next claim that each such path makes at most $\left\lfloor \sqrt{2k}\right\rfloor$ right turns.
 This follows immediately from the observation that the function $f(n) = n- \left\lceil \sqrt{2n}\right\rceil$ can be iterated on $2k$ at most $\left\lfloor \sqrt{2k}\right\rfloor$ times before reaching the value of~4 or less.
 This observation can be proved by induction, based on the fact that if $\frac{1}{2} j^2 < n \leq \frac{1}{2} (j+1)^2$, then $n- \left\lceil \sqrt{2n}\right\rceil \leq \frac{1}{2} j^2$.
 (Initially we have $j =\left\lfloor \sqrt{2k}\right\rfloor$ and finally, with $1\leq n\leq 4$, we
have $j \geq 1$.)
 As each leaf of~$T_{k,p}$ and $T_k$ is determined by the positions of the right turns on the path
leading to it from the root, we get that the number of leaves is at most $\binom{k+p}{\left\lfloor \sqrt{2k}\right\rfloor}$ in $T_{k,p}$ and at most $\binom{2k}{\left\lfloor \sqrt{2k}\right\rfloor}$ in $T_{k}$.
 The total number of nodes is therefore at most at most $2\binom{k+p}{\left\lfloor \sqrt{2k}\right\rfloor} \leq 2(k+p){\left\lfloor \sqrt{2k}\right\rfloor}$ in $T_{k,p}$ and at most $2\binom{2k}{\left\lfloor \sqrt{2k}\right\rfloor} \leq 2(2k)^{\left\lfloor \sqrt{2k}\right\rfloor}$ in $T_k$.
 As the cost of each node is at most $\mathsf{dom}_{k,p}(\left\lfloor \sqrt{2k}\right\rfloor)$ in $T_{k,p}$ and at most $\mathsf{dom}_k(\left\lceil \sqrt{2k}\right\rceil)$ in $T_{k}$, we immediately get
 \begin{eqnarray*}
   T(k,p) & \leq & 2(k+p)^{\left\lfloor \sqrt{2k}\right\rfloor} \cdot \mathsf{dom}_{k,p}(\left\lfloor \sqrt{2k}\right\rfloor) \enspace \text{ and},\\
   T(k) & \leq & 2(2k)^{\left\lfloor \sqrt{2k}\right\rfloor} \cdot \mathsf{dom}_k(\left\lceil \sqrt{2k}\right\rceil) \enspace .
 \end{eqnarray*}
\end{proof}

We obtain together with Lemma~\ref{dominionsize1} that
\begin{equation*}
  T(k,p) \leq 2(k+p)^{\left\lfloor \sqrt{2k}\right\rfloor} \cdot O\bigg(k^{\left\lfloor \sqrt{2k} \right\rfloor} \cdot T(\left\lfloor \sqrt{2k} \right\rfloor,p)\bigg),
  \end{equation*}
as well as
\begin{equation*}
  T(k) \leq 2(2k)^{\Big\lfloor \sqrt{2k}\Big\rfloor} \cdot O\left(k^{\left\lceil \sqrt{2k} \right\rceil} T(\left\lceil \sqrt{2k} \right\rceil)\right).
\end{equation*}

\begin{lemma}
  Suppose that
  \begin{equation*}
    T(k,p) \leq 2(k+p)^{\left\lfloor \sqrt{2k}\right\rfloor} \cdot O\left(k^{\left\lfloor \sqrt{2k} \right\rfloor} \cdot T(\left\lfloor \sqrt{2k} \right\rfloor,p)\right)
  \end{equation*}
  and that $T(\ell,q) \leq c' \cdot (q+\ell)^{8\left\lfloor \sqrt{2\ell}\right\rfloor}$ for some constant $c' \in \mathbb{R}$ and for all pairs $(\ell, q) \in \{1,\hdots,4\}\times\mathbb{N}$.
  Then there exist constants $c_1 \geq c',c_2\geq 8$ such that for all $k \in \mathbb{N}$,
  \begin{equation*}
    T(k,p)\leq c_1 \cdot (p+k)^{c_2\left\lfloor \sqrt{2k}\right\rfloor} \enspace .
  \end{equation*}
\end{lemma}
\begin{proof}
  Since we have $T(k,p) \leq 2(k+p)^{\left\lfloor \sqrt{2k}\right\rfloor} \cdot O\left(k^{\left\lfloor \sqrt{2k} \right\rfloor} \cdot T(\left\lfloor \sqrt{2k} \right\rfloor,p)\right)$, there exists a constant $c'_1>0$ such that 
$T(k,p) \leq 2(k+p)^{\left\lfloor \sqrt{2k}\right\rfloor} \cdot c'_1 k^{\left\lfloor \sqrt{2k} \right\rfloor} \cdot T(\left\lfloor \sqrt{2k} \right\rfloor,p)$.
  Let $\alpha_k = \frac{\left\lfloor \sqrt{2k}\right\rfloor}{\left\lfloor \sqrt{2\left\lfloor \sqrt{2k}\right\rfloor}\right\rfloor}$.
  Then for $k\geq 5$, it holds $\alpha_k \geq 1.5 > 1$.
  Let $c_1=\max\{c'_1,c'\}$, and let $c_2 = 6+ 3\log{(2c_1)}$.
  Suppose, for sake of contradiction, that the statement of the lemma does not hold for this choice of $(c_1,c_2)$.
  Then there exists a pair $(k',p')\in \mathbb{N}\times \mathbb{N}$ for which $T(k',p')> c_1 \cdot (p'+k')^{c_2\left\lfloor \sqrt{2k'}\right\rfloor}$.
  Let $k'$ be the smallest integer for which such a pair exists, and let $p' = p'(k')$ be the smallest integer for which $T(k',p')> c_1 \cdot (p'+'k)^{c_2\left\lfloor \sqrt{2k'}\right\rfloor}$.
  Note that $k' \geq 4$ and $T(\ell, q) \leq c_1 \cdot q^3 \cdot (q+\ell)^{c_2\sqrt{\ell}}$ for all pairs $(\ell,q)$ with $\ell \leq k'$, $q\leq p'$ and $\ell + q < k'+p'$.
  Further, it holds $c_2 \geq \frac{c_2}{\alpha_k}+2+ \log{(2c'_1)}$ for all $k \geq 4$.
  For $k\geq 4$ we also have $\left\lfloor \sqrt{2k} \right\rfloor <k$.
  This implies that
  \begin{align*}
    T(k',p') &\leq 2(k'+p')^{\left\lfloor \sqrt{2k'}\right\rfloor} \cdot c'_1\cdot k'^{\left\lfloor \sqrt{2k'} \right\rfloor} \cdot T(\left\lfloor \sqrt{2k'} \right\rfloor,p'))\\
			&\leq 2(k'+p')^{\left\lfloor \sqrt{2k'}\right\rfloor} \cdot c'_1\cdot k'^{\left\lfloor \sqrt{2k'} \right\rfloor} \cdot c_1 \cdot (p'+\left\lfloor \sqrt{2k'} \right\rfloor)^{c_2\left\lfloor \sqrt{2\left\lfloor \sqrt{2k'} \right\rfloor}\right\rfloor}\\
			&\leq 2(k'+p')^{\left\lfloor \sqrt{2k'}\right\rfloor} \cdot c'_1 \cdot k'^{\left\lfloor \sqrt{2k'} \right\rfloor} \cdot c_1 \cdot (p'+\left\lfloor \sqrt{2k'} \right\rfloor)^{c_2\left\lfloor \sqrt{2\left\lfloor \sqrt{2k'} \right\rfloor}\right\rfloor}\\
			&\leq (2c'_1 c_1) (k'+p')^{2\left\lfloor \sqrt{2k'} \right\rfloor+c_2\left\lfloor \sqrt{2\left\lfloor \sqrt{2k'} \right\rfloor}\right\rfloor}\\
			&\leq c_1 (k'+p')^{2\left\lfloor \sqrt{2k'} \right\rfloor+\frac{c_2}{a_{k'}}\left\lfloor \sqrt{2k'} \right\rfloor+\log{(2c'_1)}}\\
			&\leq c_1 (k'+p')^{(2+\frac{c_2}{a_{k'}}+\log{(2c'_1)})\left\lfloor \sqrt{2k'} \right\rfloor}\\
			&\leq c_1 (k'+p')^{c_2 \left\lfloor \sqrt{2k'} \right\rfloor}
  \end{align*}
  This contradicts the existence of $k'$, and therefore concludes the proof.
\end{proof}
	
\begin{lemma}
  Suppose that
  \begin{equation*}
    T(k) \leq 2(2k)^{\left\lfloor \sqrt{2k}\right\rfloor} \cdot O\left(k^{\left\lfloor \sqrt{2k} \right\rfloor} T(\left\lfloor \sqrt{2k} \right\rfloor)\right)
  \end{equation*}
  and that $T(\ell) \leq c'\ell^{\left\lfloor \sqrt{2\ell}\right\rfloor}$ for some constant $c' \in \mathbb{R}$ and for all $\ell \leq 4$.
  Then there exist constants $c_1 \geq c',c_2\geq 1$ such that for all $k \in \mathbb{N}$,
  \begin{equation*}
    T(k)\leq c_1 k^{c_2\left\lfloor \sqrt{2k}\right\rfloor} \enspace .
  \end{equation*}
\end{lemma}
\begin{proof}
  Since $T(k) \leq 2(2k)^{\left\lfloor \sqrt{2k}\right\rfloor} \cdot O\left(k^{\left\lfloor \sqrt{2k} \right\rfloor} T(\left\lfloor \sqrt{2k} \right\rfloor)\right)$, there exists a constant $c'_1>0$ such that 
  $T(k) \leq 2(2k)^{\left\lfloor \sqrt{2k}\right\rfloor} \cdot c'_1\left(k^{\left\lfloor \sqrt{2k} \right\rfloor} T(\left\lfloor \sqrt{2k} \right\rfloor)\right)$.
  Let $\alpha_k = \frac{\left\lfloor \sqrt{2k}\right\rfloor}{\left\lfloor \sqrt{2\left\lfloor \sqrt{2k}\right\rfloor}\right\rfloor}$.
  Then for $k\geq 5$ it holds that $\alpha_k \geq 1.5 > 1$.
  Let $c_1=\max\{c'_1,c'\}$ and let $c_2 = 12+ 3\log{c_1}$.
  Suppose, for sake of contradiction, that the statement of the lemma does not holds for this choice of $(c_1,c_2)$.
  Then exists a $k'\in \mathbb{N}$ such that $T(k')> c_1 k'^{c_2\left\lfloor \sqrt{2k'}\right\rfloor}$.
  Let $k'$ be the smallest such integer.
  Note that we must have $k' \geq 4$ and $T(\ell) \leq c_1 \ell^{c_2\left\lfloor \sqrt{2\ell}\right\rfloor}$ for all $\ell < k'$. 
  Further, it holds that $c_2 \geq \frac{c_2}{\alpha_k'}+4+\log{c_1}$ for all $k \geq 5$.
  For $k\geq 4$ we also have $\left\lfloor \sqrt{2k} \right\rfloor <k$.
  This implies that
  \begin{align*}
    T(k') &\leq 2(2k')^{\left\lfloor \sqrt{2k'}\right\rfloor} \cdot c'_1\left(k'^{\left\lfloor \sqrt{2k'} \right\rfloor} T(\left\lfloor \sqrt{2k'} \right\rfloor)\right)\\
     &\leq 2(2k')^{\left\lfloor \sqrt{2k'}\right\rfloor} \cdot c_1\left(2c_1 k'^{\left\lfloor \sqrt{2k'} \right\rfloor} k'^{c_2\left\lfloor \sqrt{2\left\lfloor \sqrt{2k'}\right\rfloor}\right\rfloor}\right)\\
     &\leq 2c_1(2k')^{\left\lfloor \sqrt{2k'}\right\rfloor} k'^{(\frac{c_2}{\alpha_k'}+1)\left\lfloor \sqrt{2k'}\right\rfloor+\log{c_1}}\\
     &\leq 2c_1k'^{2\left\lfloor \sqrt{2k'}\right\rfloor} k'^{(\frac{c_2}{\alpha_k'}+1)\left\lfloor \sqrt{2k'}\right\rfloor+\log{c_1}}\\
     &\leq c_1 k'^{\left\lfloor \sqrt{2k'}\right\rfloor(\frac{c_2}{\alpha_k'}+3)+\log{c_1}+\log{2}}\\
     &\leq c_1 k'^{\left\lfloor \sqrt{2k'}\right\rfloor(\frac{c_2}{\alpha_k'}+4+\log{c_1})}\\
     &\leq c_1 k'^{c_2\left\lfloor \sqrt{2k'}\right\rfloor} \enspace .
  \end{align*}
  This contradicts the existence of $k'$, and therefore concludes the proof.
\end{proof}

Since $T(k,p) \in O(p^{(k^2)})$, it holds that $O(p^{8\left\lfloor \sqrt{2k}\right\rfloor})$.
Moreover, as $T(k) = O(1)$ for $k\leq 4$, we conclude that $T(k,p) = (p+k)^{O(\sqrt{k})}$ and $T(k) = k^{O(\sqrt{k})}$.
This completes the proof of Theorem~\ref{thm:solve-recurrences1}.
\hfill$\qed$

Next, we will prove the following.
\begin{lemma}
\label{thm:solve-recurrences2}
  For $s_j,j,n\in\mathbb N$, $s = \left\lceil \sqrt{\frac{s_j}{\log_j{s_j}}}\right\rceil$ and $\ell = \left\lceil\sqrt{2(n-s_j)}\right\rceil$ we obtain
  \begin{equation*}
    T(s_j,n-s_j) = n^{O\left(\sqrt{n-s_j}+\sqrt{\frac{s_j}{\log_j{s_j}}}\right)} \enspace .
  \end{equation*}
\end{lemma}
To prove Lemma~\ref{thm:solve-recurrences2}, we first establish another lemma.
\begin{lemma}
  For $s_j,j,n\in\mathbb N$, $s = \left\lceil \sqrt{\frac{s_j}{\log_j{s_j}}}\right\rceil$ and $\ell = \left\lceil\sqrt{2(n-s_j)}\right\rceil$, it holds
  \begin{equation*}
    T(s_j,n-s_j) \leq n^{O\left(\sqrt{n-s_j}+\sqrt{\frac{s_j}{\log_j{s_j}}}\right)} \cdot \left(\mathsf{dom}_{n,n-s_j}(\ell,s)+O(1)\right) \enspace .
  \end{equation*}
\end{lemma}
\begin{proof}
  For each parity game $G$ on $n$ nodes and $s_j = s_j(G)$ we construct a binary tree $T_G$ in the following way.
  The root of $T_G$ is labeled by $(s_j,n-s_j)$.
  A node labeled by $(a,b)$ with $a > 3$ and $b > 3$ has up to two children: a left child labeled by $(a-1,b)$ or $(a,b-1)$, and possibly a right child labeled by $(a-\sqrt{a\cdot \log_j{a}},b)$ or $(a,b-\sqrt{b})$.
  Each child of a node corresponds to one of the recursive calls.
  The choice on how we label the children depends on the behavior of the algorithm.
  We label the children of a node by $(a',b')$ and $(a'',b'')$ such that the recursive calls of the algorithm are to games containing at most $a'$, respectively $a''$ nodes of out-degree at most $j$ and at most $b'$, respectively $b''$, nodes of out-degree greater than $j$.
  Nodes labeled by $(a,b)$ with $a,b \in  \{0,1, 2, 3\}$ are leaves.
  A node labeled by $(a,b)$ has a cost of $\left(\mathsf{dom}_{a+b,b}(\sqrt{b},\sqrt{a\cdot \log_j{a}})+O(1)\right)$ associated with it.

  It follows from Lemma~\ref{lemma:recurrences3} that the sum of the costs of the nodes of $T_G$ is an upper bound on the run time of $\textbf{new-win}(G,j)$.
  The worst possible sum of the costs of the nodes of $T_G$ we can obtain for some instance $G$ with $s_j = s_j(G)$ and $n = |V|$ therefore is an upper bound of $T(s_j,n-s_j)$.
  Clearly, the length of every path in~$T_G$ from the root to a leaf is at most $n$.
  We say that such a path makes a \emph{right turn} when it descends from a node to its right child.
  We next claim that each such path makes at most $O(\sqrt{n-s_j}+\sqrt{\frac{s_j}{\log_j{s_j}}})$ right turns.
  This follows immediately from the observation that the function $f(x) = x- \left\lceil \sqrt{2x}\right\rceil$ can be  iterated on $n-s_j$ at most $O(\sqrt{n-s_j})$ times before reaching the value of 3 or less and the function $g(x) = x -\left\lceil \sqrt{x\cdot\log_j{x}}\right\rceil$ can be iterated on $s_j$ at most $O(\sqrt{\frac{s_j}{\log_j{s_j}}})$ times before reaching the value of 3 or less.
  As each leaf of $T_G$ is determined by the positions of the right turns on the path
 leading to it from the root, we get that the number of leaves in $T_G$ is at most $n^{O\left(\sqrt{n-s_j}+\sqrt{\frac{s_j}{\log_j{s_j}}}\right)}$.
  The total number of nodes in $T_G$ is therefore at most $n^{O\left(\sqrt{n-s_j}+\sqrt{\frac{s_j}{\log_j{s_j}}}\right)}$.
  As the cost of each node is at most $\left(\mathsf{dom}_{n,n-s_j}(\ell,s)\right) + O(1)$, we immediately have that
  \begin{equation*}
  T(s_j,n-s_j) \leq n^{O\left(\sqrt{n-s_j}+\sqrt{\frac{s_j}{\log_j{s_j}}}\right)} \cdot \left(\mathsf{dom}_{n,n-s_j}(\ell,s)+O(1)\right).\quad\qedhere
  \end{equation*}
\end{proof}
Together with Lemma~\ref{dominionsize2}, we obtain
\begin{equation*}
  T(s_j,n-s_j) = n^{O\left(\sqrt{n-s_j}+\sqrt{\frac{s_j}{\log_j{s_j}}}\right)}.
\end{equation*}
This completes the proof of Lemma~\ref{thm:solve-recurrences2}.

\medskip
\noindent
\textbf{Acknowledgements.}
M. M. thanks L{\'a}szlo V{\'e}gh for introducing him to parity games, and the authors of~\cite{GajarskyEtAl2015} for sending us a preprint.

\bibliographystyle{abbrv}
\bibliography{paritygames}

\begin{thebibliography}{10}

\bibitem{BerwangerEtAl2006}
D.~Berwanger, A.~Dawar, P.~Hunter, and S.~Kreutzer.
\newblock D{AG}-width and parity games.
\newblock In {\em Proc. S{TACS} 2006}, volume 3884 of {\em Lecture Notes
  Comput. Sci.}, pages 524--536. 2006.

\bibitem{BerwangerEtAl2004}
D.~Berwanger and E.~Gr{\"a}del.
\newblock Fixed-point logics and solitaire games.
\newblock {\em Theory Comput. Syst.}, 37(6):675--694, 2004.

\bibitem{BerwangerEtAl2012}
D.~Berwanger, E.~Gr{\"a}del, {\L}.~Kaiser, and R.~Rabinovich.
\newblock Entanglement and the complexity of directed graphs.
\newblock {\em Theoret. Comput. Sci.}, 463:2--25, 2012.

\bibitem{BjorklundEtAl2003}
H.~Bj{\"o}rklund, S.~Sandberg, and S.~Vorobyov.
\newblock A discrete subexponential algorithm for parity games.
\newblock In {\em Proc. S{TACS} 2003}, volume 2607 of {\em Lecture Notes
  Comput. Sci.}, pages 663--674. 2003.

\bibitem{EmersonJutna1991}
E.~A. Emerson and C.~S. Jutla.
\newblock Tree automata, mu-calculus and determinacy.
\newblock In {\em Proc. FOCS 1991}, pages 368--377, 1991.

\bibitem{FearnleyLachish2011}
J.~Fearnley and O.~Lachish.
\newblock Parity games on graphs with medium tree-width.
\newblock In {\em Proc. MFCS 2011}, volume 6907 of {\em Lecture Notes Comput.
  Sci.}, pages 303--314. 2011.

\bibitem{FearnleySchewe2012}
J.~Fearnley and S.~Schewe.
\newblock Time and parallelizability results for parity games with bounded
  treewidth.
\newblock In {\em Proc. ICALP 2012}, volume 7392 of {\em Lecture Notes Comput.
  Sci.}, pages 189--200. 2012.

\bibitem{GajarskyEtAl2015}
J.~Gajarsk{\'y}, M.~Lampis, K.~Makino, V.~Mitsou, and S.~Ordyniak.
\newblock Parameterized algorithms for parity games.
\newblock In {\em Proc. MFCS 2015}, Lecture Notes Comput. Sci., pages 336--347,
  2015.

\bibitem{GradelEtAl2002}
E.~Gr\"{a}del, W.~Thomas, and T.~Wilke, editors.
\newblock {\em Automata, Logics, and Infinite Games: {A} Guide to Current
  Research}, volume 2500 of {\em Lecture Notes Comput. Sci.} Springer, 2002.

\bibitem{Jurdzinski1998}
M.~Jurdzi{\'n}ski.
\newblock Deciding the winner in parity games is in {$\rm UP\cap co$}-{$\rm
  UP$}.
\newblock {\em Inform. Process. Lett.}, 68(3):119--124, 1998.

\bibitem{Jurdzinski2000}
M.~Jurdzi{\'n}ski.
\newblock Small progress measures for solving parity games.
\newblock In {\em Proc. S{TACS} 2000}, volume 1770 of {\em Lecture Notes
  Comput. Sci.}, pages 290--301. 2000.

\bibitem{JurdzinskiEtAl2008}
M.~Jurdzi{\'n}ski, M.~Paterson, and U.~Zwick.
\newblock A deterministic subexponential algorithm for solving parity games.
\newblock {\em SIAM J. Comput.}, 38(4):1519--1532, 2008.

\bibitem{King2001}
V.~King, O.~Kupferman, and M.~Y. Vardi.
\newblock On the complexity of parity word automata.
\newblock In {\em Proc. FOSSACS 2001}, volume 2030 of {\em Lecture Notes
  Comput. Sci.}, pages 276--286. 2001.

\bibitem{KloksBodlaender1992}
T.~Kloks and H.~L. Bodlaender.
\newblock On the treewidth and pathwidth of permutation graphs, 1992.
\newblock \url{http://www.cs.uu.nl/research/techreps/repo/CS-1992/1992-13.pdf}.

\bibitem{McNaughton1993}
R.~McNaughton.
\newblock Infinite games played on finite graphs.
\newblock {\em Ann. Pure Appl. Logic}, 65(2):149--184, 1993.

\bibitem{Obdrzalek2003}
J.~Obdr{\v z}{\'a}lek.
\newblock Fast $\mu$-calculus model checking when tree-width is bounded.
\newblock In {\em Proc. CAV 2003}, volume 2725 of {\em Lecture Notes Comput.
  Sci.}, pages 80--92. 2003.

\bibitem{Obdrzalek2007}
J.~Obdr{\v{z}}{\'a}lek.
\newblock Clique-width and parity games.
\newblock In {\em Proc. CSL 2007}, volume 4646 of {\em Lecture Notes Comput.
  Sci.}, pages 54--68. 2007.

\bibitem{Schewe2007}
S.~Schewe.
\newblock Solving parity games in big steps.
\newblock In {\em Proc. FSTTCS 2007}, volume 4855 of {\em Lecture Notes Comput.
  Sci.}, pages 449--460. 2007.

\bibitem{Stirling1995}
C.~Stirling.
\newblock Local model checking games.
\newblock In {\em Proc. CONCUR 1995}, volume 962 of {\em Lecture Notes Comput.
  Sci.}, pages 1--11. 1995.

\bibitem{JurdzinskiVoge2000}
J.~V{\"o}ge and M.~Jurdzi{\'n}ski.
\newblock A discrete strategy improvement algorithm for solving parity games.
\newblock In {\em Proc. CAV 2000}, volume 1855 of {\em Lecture Notes Comput.
  Sci.}, pages 202--215. 2000.

\bibitem{Zielonka1998}
W.~Zielonka.
\newblock Infinite games on finitely coloured graphs with applications to
  automata on infinite trees.
\newblock {\em Theoret. Comput. Sci.}, 200(1-2):135--183, 1998.

\end{thebibliography}

\end{document}